\documentclass[11pt,a4paper]{article}

\usepackage[top=1in,bottom=1in,right=1in,left=1in]{geometry}
\usepackage[T1]{fontenc}
\usepackage{luainputenc}

\usepackage{epsfig, float}
\usepackage{xfrac}

\usepackage{enumitem}
\usepackage{float}
\usepackage{lscape}
\usepackage{setspace} 
\usepackage{graphicx}
\usepackage{mathcomp}
\usepackage{dsfont}
\usepackage{lmodern}

\usepackage{todonotes}

\usepackage{amsmath,amssymb}
\usepackage{amsfonts}
\usepackage{url}
\usepackage{bbm}
\usepackage{mathrsfs} 
\usepackage[english]{babel}
\usepackage[hidelinks]{hyperref}

\usepackage{array}

\newcommand{\R}{\mathbb{R}}
\newcommand{\C}{\mathbb{C}}
\newcommand{\N}{\mathbb{N}}

\newcommand{\F}{\mathscr{F}}
\newcommand{\id}{\mathbbm{1}}
\newcommand{\supp}{{\rm supp} \, }

\newcommand{\dom}{{\rm dom}}

\newcounter{deefinition}
\newtheorem{theorem}{Theorem} 
\newtheorem{lemma}[theorem]{Lemma}

\newtheorem{corollary}[theorem]{Corollary}

\newtheorem{definition}[deefinition]{Definition}

\newenvironment{proof}[1][Proof:]{\begin{trivlist}
\item[\hskip \labelsep {\bfseries #1}]}{\end{trivlist}}
\newenvironment{proofU}[1][Proof of Theorem \ref{thm:uniqueness}:]{\begin{trivlist}
\item[\hskip \labelsep {\bfseries #1}]}{\end{trivlist}}

\newenvironment{remark}[1][Remark:]{\begin{trivlist}
\item[\hskip \labelsep {\bfseries #1}]}{\end{trivlist}}

\newcommand{\qed}{\hfill\ensuremath{\square}}
\newcolumntype{C}[1]{>{\centering\let\newline\\\arraybackslash\hspace{0pt}}m{#1}}

\begin{document}

\title{\Large\textsc{Multi-time dynamics of the\\
Dirac-Fock-Podolsky model of QED}}

\author{Dirk-Andr\'e Deckert\thanks{deckert@math.lmu.de} \ and  Lukas
Nickel\thanks{nickel@math.lmu.de} \\[0.2cm] Mathematisches Institut,
Ludwig-Maximilians-Universit\"at\\ Theresienstr.\ 39, 80333 M\"unchen, Germany
}

\date{\today}

\maketitle

\begin{abstract} 
    Dirac, Fock, and Podolsky \cite{DiracFockPodolsky} devised a relativistic
    model in 1932 in which a fixed number of $N$ Dirac electrons interact
    through a second-quantized electromagnetic field. It is formulated with the
    help of a multi-time wave function $\psi(t_1,\mathbf{x}_1,...,t_N,
    \mathbf{x}_N)$ that generalizes the Schr\"odinger multi-particle wave function
    to allow for a manifestly relativistic formulation of wave mechanics. The
    dynamics is given in terms of $N$ evolution equations that have to be solved
    simultaneously.  Integrability imposes a rather strict constraint on the
    possible forms of interaction between the $N$ particles and makes the
   rigorous construction of interacting dynamics a long-standing problem, also present
    in the modern formulation of quantum field theory. 
    For a simplified version
    of the multi-time model, in our case describing $N$ Dirac electrons that
    interact through a relativistic scalar field, we prove well-posedness of the
    corresponding multi-time initial value problem and discuss the mechanism and
    type of interaction between the charges. For the sake of mathematical rigor
    we are forced to employ an ultraviolet cut-off in the scalar field. Although
    this again breaks the desired relativistic invariance, this violation occurs
    only on the arbitrary small but finite length-scale of this cut-off. In view
    of recent progress in this field, the main mathematical challenges faced in
    this work are, on the one hand, the unboundedness from below of the free
    Dirac Hamiltonians and the unbounded, time-dependent interaction terms, and
    on the other hand, the necessity of pointwise control of the multi-time wave
    function. \\

\end{abstract}

\section{Introduction}

\subsection{The need for multi-time models}

The multi-time formalism for relativistic wave mechanics was first developed in
works of Dirac \cite{DiracRQM, DiracFockPodolsky} and Bloch
\cite{blochmultitime} and after Tomonaga's famous paper \cite{tomonaga}
ultimately lead towards the modern relativistic formulation of QFT. At its
base, the main observation is that the Schr\"odiger wave function for a many-body
system contains only one time variable $t$ and $N$ position variables
$\mathbf{x}_i$, $i=1,\ldots,N$, in other words a configuration of $N$
space-time coordinates $(t,\mathbf{x}_i)$, $i=1,\ldots,N$, on an equal-time
hypersurface ${t}\times\mathbb R^3$ in Minkowski space. A Lorentz-boost will in
general lead to a configuration of space-time points $(t_i',\mathbf{x}_i')$,
$i=1,\ldots,N$, with pair-wise distinct $t_i',t_j'$, hence, a Schr\"odinger
wave function defined on equal time hypersurfaces will fail to have the desired
transformation properties under Lorentz boosts. A natural way to extend the wave function on equal-time hypersurfaces is the multi-time wave function $\psi(t_1,\mathbf{x}_1,...,t_N,\mathbf{x}_N)$, an object which lives on a subset of $\mathbb R^{4N}$.

In recent years, there has been a renewed interest in constructing mathematically
rigoros multi-time models, see \cite{multitimegeneral} for an overview. Some of the current efforts to understand Dirac's multi-time models focus on the well-posedness of the corresponding initial value problems \cite{multitimeNoPotentials, 1d_model, NTeilchenPaper, neurodi, markuss},
other works also ask the question how the multi-time formalism could be exploited to avoid the infamous ultraviolet divergence of relativistic QFT and how a varying number of particles by means of creation and annihilation processes can be addressed
\cite{multi-timeQFT, multitimePaircreation, MultitimeIBC}. Beside
being candidate models for fundamental formulations of relativistic wave
mechanics, a better mathematical understanding of such multi-time evolutions
may also be beneficial regarding more technical discussions, such as the control of
scattering estimates on vacuum expectation values of products of interacting
field operators; see e.g.\ \cite{dybalski}.

Many contemporary treatments of multi-time models are yet not entirely
satisfactory as they either have technical deficiencies, e.g., do not allow to
treat unbounded Hamiltonians, or define interactions whose nature are
conceptually not entirely clear or experimentally adequate. Also our treatment
presented in this work is not fully satisfactory by those standards, as for the
sake of mathematical rigor we need to introduce an ultraviolet cut-off that in
turn breaks the Lorentz-invariance of the model. Nevertheless, building on
previous works, we still achieve a substantial improvement since we can allow
for unbounded Hamiltonians in the evolution equations. Furthermore, the
violation of Lorentz-invariance only occurs on the finite but arbitrary small
length-scale of the cut-off. Since the mathematically rigorous treatment of
multi-time evolutions is independent of the ultraviolet divergences of
relativistic interaction, we believe that it is advantageous for the progress in
both topics to separate the discussion between formulations of multi-time
dynamics and the divergences of quantum field theory at first. Later, it may well
be that the understanding of multi-time evolution leads to new possibilities to
encode relativistic interaction without causing ultraviolet divergences. \\

This work is divided into three parts. First, we give an informal introduction
to the model at hand in subsection \ref{sec:motivationdfp}. The mathematical
definition of this model is then given in section \ref{sec:defdfpmodel} where
we state our main results on existence, uniqueness, and interaction of
solutions, i.e., Theorem~\ref{thm:mainthm}, Theorem~\ref{thm:uniqueness}, and
Theorem~\ref{thm:interaction}, respectively. The corresponding proofs are
provided in section \ref{sec:theproofs}.  

\subsection{The multi-time model} \label{sec:motivationdfp} 

In our choice of model we follow closely the Dirac, Fock, Podolsky (DFP) model given in the paper
\cite{DiracFockPodolsky}, which we informally introduce in this
subsection and formally define in the next one. This model is supposed to
describe the relativistic interaction between $N$ persistent Dirac electrons.
The only simplification we assume for the model treated in this paper in
comparison to the original DFP model is that the electromagnetic
interaction is replaced by the one of a scalar field. This allows to avoid the
additional complication of electromagnetic gauge freedom. A ready
choice for the evolution equations of the multi-time wave function
$\psi(x_1,...,x_N)$ is a system of $N$ Hamiltonian equations,
\begin{equation} \label{eq:firstthing}
    i \partial_{t_j} \psi(x_1,...,x_N) = \mathcal{H}_j \psi(x_1,...,x_N), \quad j=1,...,N,
\end{equation}
with a suitable \textit{partial Hamiltonian} $\mathcal{H}_j$ for each particle.
In \cite{blochmultitime}, Bloch argued that it is necessary for the existence
of solutions to \eqref{eq:firstthing} that an integrability condition for the
different times $t_j$, the so-called \textit{consistency condition}
\begin{equation} \label{eq:consistencycond}
    \left[ \mathcal{H}_j, \mathcal{H}_k \right] + i \frac{\partial
    \mathcal{H}_j}{\partial t_k} - i \frac{\partial \mathcal{H}_k}{\partial
t_j} = 0, \quad \forall j \neq k,
\end{equation}
is satisfied in the domain of $\psi$, which is usually taken as the set of
space-like configurations in $\R^{4N}$.

 Let $\mathcal{H}^0_j = -i \gamma^0_j \boldsymbol{\gamma}_j \cdot
\boldsymbol{\nabla}_j + \gamma^0_j m$ be the free Dirac operator acting on particle $j$, with the usual gamma matrices $\gamma^\mu_j$. For the free multi-time evolution with Hamiltonians $\mathcal{H}_j =
\mathcal{H}^0_j$,
condition \eqref{eq:consistencycond} is fulfilled.  For the introduction of a non-trivial interaction, however, the consistency condition poses a serious obstacle. If one takes as partial Hamiltonians 
\begin{equation}
    \mathcal{H}_j = \mathcal{H}^0_j + V_j(x_1,...,x_N),
\end{equation}
with interaction potentials, i.e.\ multiplication operators, $V_j$, it is
hardly possible to fulfill \eqref{eq:consistencycond}. Using this insight, it
was shown in \cite{multitimeNoPotentials, consistencynickeldeckert} that
systems of multi-time Dirac equations with relativistic interaction potentials
fail to admit solutions.

Already in 1932 in \cite{DiracRQM}, Dirac pointed out an ingenious way to
circumvent this problem, namely, by second quantization. He observed that in
case the ``potential''  is not a multiplication operator, but a 
Fock space valued field operator $\varphi(x)$, the consistency condition
\eqref{eq:consistencycond} can be retained although it will turn out
that an interaction is present. The Hamiltonians in question are of the form
\begin{equation} \label{eq:Hamiltonian}
    \mathcal{H}_j = \mathcal{H}^0_j + \varphi(t_j, \mathbf{x}_j),
\end{equation}
all containing one and the same second quantized scalar field $\varphi$ on
space-time $\R^4$, fulfilling the wave equation
\begin{equation}
    \label{eq:fieldeqns}
    \square_{x} \varphi(x) = \left(\partial_t^2 - \triangle_{\mathbf{x}} \right)
    \varphi(t,\mathbf{x})   = 0,
\end{equation}
as well as the canonical commutation relation
\begin{equation} \label{eq:CCR}
    \left[ \varphi(x_j), \varphi(x_k) \right] = i \Delta(x_j,x_k),
\end{equation}
with $\Delta$ being the Pauli-Jordan function \cite{tomonaga, schwinger} given
in \eqref{eq:Paulijordan}. It is well-known that \eqref{eq:CCR} implies
\begin{equation}
    \left[ \varphi(x_j), \dot{\varphi}(x_k) \right]_{t_j=t_k} = i \delta^{(3)}(\mathbf{x}_j - \mathbf{x}_k) .
\end{equation}
This ensures the consistency of the system of equations in the sense of
\eqref{eq:consistencycond} since 
\begin{equation} \label{eq:drr}
    \Delta(x_j, x_k) = 0 \quad \text{if} \ x_j, x_k \ \text{are space-like related}.
\end{equation}
A natural choice for a representation of the field operator fulfilling
\eqref{eq:CCR} is the one on standard Fock space.  The multi-time wave-function
$\psi(x_1,...,x_N)$ can then be thought of as taking values in a bosonic Fock
space. This second quantization of $\varphi(x)$ is the key feature to
understand how the seemingly ``free'' evolutions in \eqref{eq:fieldeqns} in
fact allow to mediate interaction between the Dirac electrons. In fact, an
informal computation (see \cite{DiracFockPodolsky}) shows that \eqref{eq:fieldeqns} and \eqref{eq:drr} imply
for the field operator $\varphi_H(x):=\varphi_H(t,\mathbf{x})=U(t)^\dagger
\varphi(0,\mathbf{x}) U(t)$, where $U(t)$ denotes the time evolution of the
$N$-body system on equal-time hypersurfaces, that
\begin{equation} 
    \label{eq:ehrenfestdelta}
    \left(\partial_t^2 - \triangle_{\mathbf{x}} \right)\varphi_H(t,\mathbf{x})
    = - \sum_{j=1}^N \delta^{(3)}
    \left(\hat{\mathbf{x}}_j(t) - \mathbf{x}\right), 
\end{equation}
where $\hat{\mathbf{x}}_j(t)=U(t)^\dagger \hat{\mathbf{x}}_jU(t)$ denotes the
position operator of the $k$-th electron in the Heisenberg picture. The
right-hand side of \eqref{eq:ehrenfestdelta} now demonstrates the effective
source terms influencing the scalar field which in turn couples the motion of
the $N$ electrons. A rigoros version of this informal computation is given as
Theorem~\ref{thm:interaction}.

\paragraph{Mathematical challenges.} There are three main difficulties we have
to overcome for a mathematical solution theory of the model.  
\begin{enumerate}
    \item As it is well-known \cite{schweber}, the scalar field model is
        badly ultraviolet divergent. A standard way to defer the discussion of
        this problem and nevertheless continue the mathematical discussion is
        the introduction of a ultraviolet cut-off in the scalar field.  This
        cut-off, which can be thought of as smearing out the scalar field with
        a smooth and compactly supported function $\rho$ with diameter
        $\delta>0$, ensures well-definedness of the model, however, breaks
        Lorentz on the length scale $\delta$ as it smears out the right-hand
        side of \eqref{eq:drr} as can be seen from
        \eqref{eq:interactioninourmodel} below.  This will
        furthermore force us to take as domain $\mathscr{S}_\delta$, defined in
        \eqref{eq:defSdelta} below, for the multi-time wave function instead of
        all space-like configurations on $\mathbb R^{4N}$. Since
        $\mathscr{S}_\delta$ is not an open set in $\R^{4N}$ a
        simple notion of differentiability is not sufficient anymore which is
        reflected in our choice of solution sense in
        Definition~\ref{def:solution}.
    \item We need sufficient regularity in the
        solution candidates to allow for point-wise evaluation. It is decisive
        for our proofs that we find a dense set $\mathscr{D}$ of smooth
        functions which is left invariant by the single-time evolutions.
        Furthermore, the majority of methods employed in the
        literature on Schr\"odinger Hamiltonians (see e.g.\
        \cite{Hasler2008}) rely on boundedness from
        below, and hence, do not apply to our setting as the free Dirac
        Hamiltonian is not bounded from below.   
    \item
        Since we add unbounded and time-dependent interaction terms to the free
        Dirac Hamiltonians, already the study of the corresponding single-time
        equations generated by the Hamiltonians $\mathcal{H}_j(t)$ in
        \eqref{eq:Hamiltonian} is subtle. Abstract theorems such as the one of
        Kato \cite{katopaper} or Yosida \cite[ch.\ XIV]{yosida}  about the
        existence of a propagator $U(t,s)$ require time-independence of the
        domain $\dom(\mathcal{H}_j(t))$, which in our case is unknown. 
\end{enumerate}

Beside the introduction of an ultraviolet cut-off, which will be defined in the
next section, there is a further difference compared to the original
formulation of Dirac, Fock, Podolsky, namely that the multi-time wave function
$\psi$ of $N$ particles has $N$ time arguments and not an additional ``field
time'' argument.  This is because we formulate the field degrees of freedom in
momentum space and in the Dyson picture, leading to a time-dependent
$\varphi(t,\mathbf{x})$ but no free field Hamiltonian in $\mathcal{H}_j$. The
choice of a field time as in \cite{DiracFockPodolsky} corresponds to choosing a
space-like hypersurface $\Sigma$ (in that paper, only equal-time hypersurfaces
$\Sigma_t$ are considered) on which the field degrees of freedom are evaluated.
Our formulation is mathe\-ma\-ti\-cal\-ly convenient since the Hilbert space is
fixed and not hypersurface-dependent. It is always possible to choose a
hypersurface and perform the Fourier transformation to obtain field modes in
position space.  

\section{Definition of the model and main results} \label{sec:defdfpmodel}

We now put the model described by the informal equations \eqref{eq:firstthing},
\eqref{eq:Hamiltonian}, \eqref{eq:CCR} into a mathematical rigoros context and
define a solution sense, see Def.\ \ref{def:solution} below, which will allow
us to formulate our main results about existence and uniqueness of solutions.
As the model describes the interaction of $N$ electrons with a scalar field, an
operator on Fock space, there are two main ingredients we need to define:
the field operator and the multi-time evolution equations.

\paragraph{Field operator with Cut-off.}  We follow the standard  quantization
procedure. The Fock space is constructed by means of a direct sum
of symmetric tensor products of the one-particle Hilbert space $L^2(\R^3,\C)$
of complex valued square integrable functions on $\R^3$:
\begin{equation}
    \label{eq:fockspace}
    \F = \bigoplus_{n=0}^\infty L^2({\R^3},\C)^{\odot n},
\end{equation}
where  $\odot$ denotes the symmetric tensor product.  In our setting, we think
of $\R^3$ as momentum space.  The total Hilbert space, in which the wave
function $\psi(t_1,\cdot,...,t_N,\cdot)$ is contained for fixed time
$t_1,\ldots,t_N\in\R$, is given by
\begin{equation}
    \label{eq:hilbertspaceidentification}
    \mathscr{H} = L^2(\R^{3N}, \F^{K}) \cong L^2(\R^{3N},\C^K) \otimes \F  \cong
    L^2(\R^{3N},\C) \otimes \F^K,
\end{equation}
with $K = 4^N$ denoting the dimension of spinor space of the $N$ Dirac
electrons.  
In view of \eqref{eq:fockspace} and \eqref{eq:hilbertspaceidentification},
we use the notation
\begin{equation} 
    \begin{split}
        \text{for a.e.} \ (x_1,...,x_N): \psi(x_1,...,x_N) = \left(
            \psi^{(n)}(x_1,...,x_N) \right)_{n \in \N_0}, \\ \text{so that} \ \left(
            (\mathbf{k}_1,...,\mathbf{k}_n) \mapsto
        \psi^{(n)}(x_1,...,x_N;\mathbf{k}_1,...,\mathbf{k}_n) \right) \in \C^K
        \otimes L^2({\R^3},\C)^{\odot n}
    \end{split}
\end{equation}
to denote the $n$-particles sectors of Fock space $\F$
and distinguish between functions with values in $\F$ and $\C^K$.
A dense set in $\F$ are the finite particle vectors
$\F_{\operatorname{fin}}$. On this set,
we can define for square integrable $f$, as in Nelson's paper \cite{nelson},
the annihilation
\begin{equation}
    \left( \int d^3\mathbf{k} ~ f(\mathbf{k}) a(\mathbf{k}) \psi \right)^{(n)} (\mathbf{k}_1,...,\mathbf{k}_n) = \sqrt{n+1} \int d^3\mathbf{k} ~ f(\mathbf{k}) \psi^{(n+1)} (\mathbf{k},\mathbf{k}_1,...,\mathbf{k}_n) 
\end{equation}
and creation operators
\begin{equation}
    \left( \int d^3\mathbf{k} ~ f(\mathbf{k}) a^\dagger(\mathbf{k}) \psi \right)^{(n)} (\mathbf{k}_1,...,\mathbf{k}_n)  = \frac{1}{\sqrt{n}} \sum_{j=1}^n f(\mathbf{k}_j) \psi^{(n-1)} (\mathbf{k}_1,...,\widehat{\mathbf{k}_j},...,\mathbf{k}_n),
\end{equation}
in which a variable with hat is omitted. The field mass is $\mu \geq 0$ and the
energy $\omega(\mathbf{k}) = \sqrt{\mathbf{k}^2 + \mu^2}$, which allows to
define the free field Hamiltonian
\begin{equation}
    \left( \mathcal{H}_f \psi \right)^{(n)}(\mathbf{k}_1,...\mathbf{k}_n) = \sum_{j=1}^n \omega(\mathbf{k_j}) \psi(\mathbf{k}_1,...,\mathbf{k}_n),
\end{equation}
as self-adjoint operator on its domain $\dom(\mathcal{H}_f) \subset \F$; see
\cite{spohn}. 
We will later use the notation $\dom(\mathcal{H}_f^\infty):=
\bigcap_{j=0}^\infty \dom(\mathcal{H}_f^j)$.\\ 

Before we can define the scalar field, we need to introduce the cut-off as
final ingredient. Let $B_r(\mathbf{x})$ denote the open ball in $\R^3$ of
radius $r$ around $\mathbf{x}\in\R^3$. For this we introduce a smooth and
compactly supported real-valued 
function
\begin{equation} \label{eq:cutoff}
    \rho \in C^\infty_c(\R^{3},\R) \ \text{such that} \ \supp(\rho) \subset
    B_{\delta/2}(\mathbf{0}) ,
\end{equation} 
which can later be thought of as smearing out the point-like interaction to be
mediated by the scalar field by a charge form factor $\rho$.  The Fourier
transform $\hat{\rho}(\mathbf{k})$ is an element of the Schwartz space of
function of rapid decay with not necessarily compact support. For each particle
index $j=1,...,N$, we can now define the time-dependent scalar field
\begin{equation} \label{eq:definitionoffield}
    \varphi_j(t) \psi := \int d^3\mathbf{k} ~ \left[ \left(
            \frac{\hat{\rho}(\mathbf{k})}{\sqrt{\omega(\mathbf{k})}} e^{-i
            \omega(\mathbf{k}) t} e^{i \mathbf{k} \cdot \hat{\mathbf{x}}_j}  a(\mathbf{k})
            + \frac{\hat{\rho}^\dagger(\mathbf{k})}{\sqrt{\omega(\mathbf{k})}}  e^{i
            \omega(\mathbf{k}) t} e^{-i \mathbf{k} \cdot \hat{\mathbf{x}}_j}
    a^\dagger(\mathbf{k}) \right) \psi \right]
\end{equation}
for sufficiently regular $\psi\in\mathscr{H}$.  Here, $\hat{\mathbf{x}}_j$ is
the position operator of the $j$-th particle which acts on a multi-time wave
function by $\hat{\mathbf{x}}_j \psi(t_1,\mathbf{x}_1,...,t_j, \mathbf{x}_j,
...) = \mathbf{x}_j \psi(t_1,\mathbf{x}_1,...,t_j, \mathbf{x}_j, ...)$. 
The necessity of the cut-off function $\rho\in  C^\infty_c(\R^3)$ can be
seen from the fact that if we had chosen $\rho(\mathbf{x})=\delta^3(\mathbf{x})$
which for reasons of Lorentz invariance would be physically desirable but would
imply $\hat{\rho}= (2\pi)^{-3/2}$, the domain of the second summand in
$\varphi_j(t)$ would be $\{0\}$, which is a manifestation of the mentioned
ultraviolet problem. With a square integrable $\hat{\rho}$, the field operator
is self-adjoint on a dense domain; see \cite{spohn}. An equivalent definition
is possible by direct fiber integrals, see \cite{zenkstockmeyer, arai}. Despite
the notation, one should not think of the $\varphi_j$ as being $N$ different
fields, the index just denotes in a brief way that the single scalar field is
evaluated at the coordinates of particles $j$, i.e.\ at $\mathbf{x}_j$.

This allows to define the one-particle Hamiltonians as follows:
\begin{equation} \label{eq:DFPpsi}
    \mathcal{H}_j(t) = \mathcal{H}^0_j + \varphi_j(t), \quad j=1,...,N.
\end{equation}

\paragraph{Multi-Time Evolution Equations and Solution Sense.} 
As domain for our multi-time wave function on configuration space-time, we take
those configurations of space-time points which are at equal times or have a
space-like distance of at least $\delta$, i.e. 
\begin{equation} \label{eq:defSdelta}
\mathscr{S}_\delta := \left\lbrace \left. (x_1,...,x_N) \in \R^{4N}\right| \forall j \neq k: t_j = t_k \ \text{or} \ \|\mathbf{x}_j - \mathbf{x}_k\| > |t_j-t_k| + \delta  \right\rbrace.
\end{equation}
The multi-time wave function will hence be represented as a map $\psi:
\mathscr{S}_\delta \to \F^K$.

The natural notion of a solution to our multi-time system \eqref{eq:firstthing}
would be a smooth function mapping from $\mathscr{S}_\delta$ to the Fock space
$\F^K$. However, the above introduced Hilbert space $\mathscr{H}$ on $\R^{3N}$
allows to apply on a lot of functional analytic methods, and thus, simplifies
the mathematical analysis considerably. This is why it is helpful to at first
define a solution as a map $\psi: \R^N \to \mathscr{H}, (t_1,...,t_N) \mapsto
\psi(t_1,...,t_N)$ and require it to solve the system
\eqref{eq:firstthing} on the space-time configurations in $\mathscr{S}_\delta$.
The latter involves the difficulty that the domain $\mathscr{S}_\delta$ is not
an open set in $\R^{4N}$ so that partial derivatives with respect to time
coordinates cannot be straightforwardly defined in this set. 

\begin{figure}[h] \begin{center} 
\includegraphics[scale=1.1]{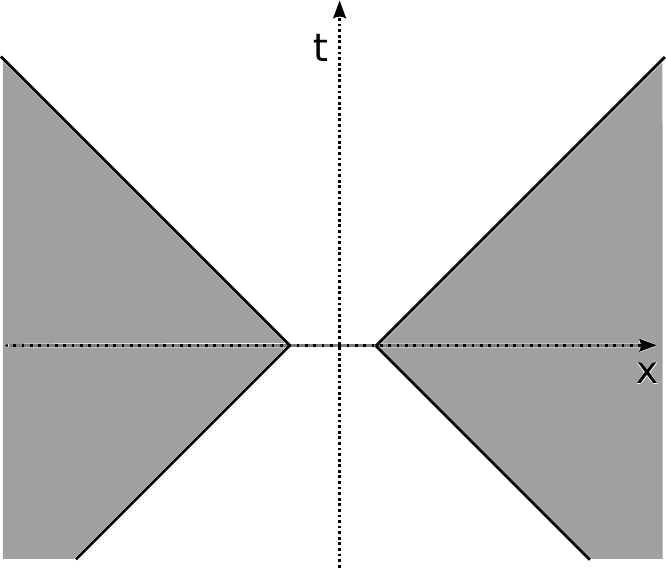}
\caption{The set $\mathscr{S}_\delta$ is depicted in grey, for two particles in relative coordinates. Because of the line at $t=t_1-t_2=0$, this is obviously not an open set in configuration space-time. At the origin, for example, the partial derivative $\partial_{t_1}$ cannot be computed inside the set.} \label{fig:Sdelta}
\end{center} \end{figure}

In order to cope with this difficulty, we adapt a method to define partial derivatives in $\mathscr{S}_\delta$ that
was also employed by Petrat and Tumulka \cite[sec.\ 4]{multitimeNoPotentials}.
If all times are pair-wise different, the usual partial derivatives exist.
However, this is not the case at points where for some $j \neq k$, $t_j = t_k$
while $\| \mathbf{x}_j - \mathbf{x}_k \| \leq \delta$. For those configurations
we will only take the derivative with respect to the common time
coordinate. This is implemented as follows: Each point
$x=(x_1,...,x_N) \in \mathscr{S}_\delta$ defines a partition of $\{ 1,...,N\}$
into non-empty disjoint subsets $P_1,...,P_L$ by the equivalence relation that
is the transitive closure of the relation that holds between $j$ and $k$
exactly if\footnote{This gives exactly the partition called $FP_{q^4}$ by
Petrat and Tumulka.} $\|\mathbf{x}_j - \mathbf{x}_k\| \leq |t_j-t_k| + \delta$.
We call this the \textit{corresponding partition} to $x$. By
\eqref{eq:defSdelta}, all particles in one set $P_i$ of the partition
necessarily have the same time coordinate, i.e.\ $\forall i \in \{1,...,L\} ~
\forall j,k \in P_i$, we have $t_j=t_k$. We write this common time coordinate
as $t_{P_i}$ for each $i=1,...,L$. 

The partial derivative with respect to $t_{P_i}$ can now be defined for a
differentiable function $\psi: \R^N \to \mathscr{H}$ as
\begin{equation} \label{eq:partialder}
\left( \frac{\partial}{\partial t_{P_i}} \psi(t_1,...,t_N)\right)(\mathbf{x}_1,...,\mathbf{x}_N) := \sum_{j \in P_i} \left( \frac{\partial}{\partial t_j} \psi(t_1,...,t_N)\right) (\mathbf{x}_1,...,\mathbf{x}_N),
\end{equation}
provided that the expression on the right-hand side is well-defined. By this
definition, $\frac{\partial}{\partial t_{P_i}} \psi$ can be obtained solely by
limits of sequences of configurations inside $\mathscr{S}_\delta$,  so changing
the function $\psi$ outside of the relevant domain $\mathscr{S}_\delta$ will
not matter for the derivative, and thus not affect its status of being a
solution. With this notation at hand, we define:

\begin{definition}[Solution Sense] \label{def:solution} 
    For each set $A \subset \{1,...,N\}$, define the respective Hamiltonian
\begin{equation} \label{eq:defHA}
\mathcal{H}_A(t) := \sum_{j \in A} \left( \mathcal{H}^0_j +  \varphi_j(t) \right).
\end{equation}
A \emph{solution of the multi-time system} is a function $\psi: \R^N \to \mathscr{H}, (t_1,...,t_N) \mapsto \psi(t_1,...,t_N)$ such that the following hold:
\begin{enumerate}
    \item[i)] Time derivatives: $\psi$ is differentiable.
    \item[ii)] Pointwise evaluation: For every
        $(t_1,\mathbf{x}_1,...,t_N,\mathbf{x}_N) \in \mathscr{S}_\delta$, and
        for all $j=1,...,N$, the following pointwise evaluations are
        well-defined:
        \begin{equation} \label{eq:pointwiseev} \begin{split}
            \Big(\psi(t_1,...,t_N)\Big)(\mathbf{x}_1,...,\mathbf{x}_N), \\ \left(\partial_{t_j} \psi(t_1,...,t_N)\right)(\mathbf{x}_1,...,\mathbf{x}_N), \\ \Big(\mathcal{H}_j(t_j) \psi(t_1,...,t_N)\Big)(\mathbf{x}_1,...,\mathbf{x}_N).
        \end{split} 
    \end{equation}
\item[iii)] Evolution equations: For every $x =
    (t_1,\mathbf{x}_1,...,t_N,\mathbf{x}_N) \in \mathscr{S}_\delta$ with
    corresponding partition $P_1,...,P_L$,  the equations
    \begin{equation} \label{eq:systemstrongsolution}
        \left( \frac{\partial}{\partial t_{P_j}} \psi(t_1,...,t_N)\right)(\mathbf{x}_1,...,\mathbf{x}_N)=\Big( \mathcal{H}_{P_j} (t_{P_j})  \psi \left(t_1,...,t_N \right) \Big) (\mathbf{x}_1,...,\mathbf{x}_N) , \quad j=1,...,L,
    \end{equation}
    where the left hand side is defined by \eqref{eq:partialder}, are
    satisfied.
\end{enumerate}
\end{definition}

Due to the unfamiliar structure of the domain $\mathscr{S}_\delta$ and our
compact notation, this definition may look complicated at first sight. However,
the complication is only due to the introduction of the cut-off $\rho$ which
led to the definition of $\mathscr{S}_\delta$. The purpose of the whole effort is simply 
to restrict the system \eqref{eq:firstthing} to those time
directions in which taking the derivative is admissible in
$\mathscr{S}_\delta$. It may be helpful to take a quick look at
Eq.\ \eqref{eq:syys} which shows the explicit form of the multi-time system for
the special case of $N=2$. We emphasize that with our notation in
\eqref{eq:defHA}, the index of the Hamiltonian is actually a set, for example
$\mathcal{H}_{\{1,2\}} = \mathcal{H}_{1} + \mathcal{H}_{2}$ denoting the mutual
Hamiltonian of particles $1$ and $2$. 

As a final ingredient, we define a dense domain in $\mathscr{H}$:
\begin{equation} \label{eq:defD}
    \mathscr{D} := C^\infty_c(\R^{3N},\C^K) \otimes \F \cap L^2(\R^{3N}, \C^K) \otimes \dom(\mathcal{H}_f^\infty).
\end{equation}
Our first main result is on the existence of solutions given initial values in
$\mathscr{D}$.
\begin{theorem}[Existence] \label{thm:mainthm}     
   Let $\psi^0 \in \mathscr{D}$. Then there is a solution of the multi-time
    system $\psi$ in the sense of definition \ref{def:solution} which satisfies
    $\psi(0,...,0) = \psi^0$ pointwise. In particular, there is such a solution
    $\psi$ fulfilling
    \begin{equation} \label{eq:psibarisgood}
        \psi (t_1, \cdot, ..., t_N, \cdot) \in \mathscr{D} \quad \forall (t_1,...,t_N) \in \R^N.
    \end{equation}
\end{theorem}
The second main result is on the uniqueness of solutions in $\mathscr{D}$.
\begin{theorem}[Uniqueness] \label{thm:uniqueness}
    Let $\psi^0 \in \mathscr{D}$. Let $\psi_1$ and $\psi_2$ be two solutions of
    the multi-time system in the sense of definition \ref{def:solution} which
    both satisfy $\psi_k(0,...,0) = \psi^0$ pointwise for $k=1,2$.  Then we
    have for all $(t_1,\mathbf{x}_1,...,t_N,\mathbf{x}_N) \in
    \mathscr{S}_{\delta}$:
    \begin{equation}
        \Big( \psi_1(t_1,...,t_N) \Big) (\mathbf{x}_1,...,\mathbf{x}_N) = \Big( \psi_2(t_1,...,t_N) \Big) (\mathbf{x}_1,...,\mathbf{x}_N).
    \end{equation}
\end{theorem}
To illustrate that our model is indeed interacting, we provide a rigoros
version of Eq.\ \eqref{eq:ehrenfestdelta} for the case of our model, in other words, the Ehrenfest equation for the scalar field operator.

\begin{theorem}\label{thm:interaction} 
    For every $t \in \R$ and $\mathbf{x} \in \R^3$, let us abbreviate the
    solution to given initial values $\psi^0 \in \mathscr{D}$ at equal times as
    $\psi^t:=U_{\{1,...,N\}}(t,0)\psi^0$ and
    $\mathcal{H}^t:=\mathcal{H}_{\{1,...,N\}}(t)$ and write
    $\varphi(t,\mathbf{x})$ for the field operator acting as 
    \begin{equation}
        \label{eq:definitionoffieldx} \varphi(t,\mathbf{x}) \psi := \int
        d^3\mathbf{k} ~ \left[ \left(
        \frac{\hat{\rho}(\mathbf{k})}{\sqrt{\omega(\mathbf{k})}} e^{-i
\omega(\mathbf{k}) t} e^{i \mathbf{k} \cdot \mathbf{x}}  a(\mathbf{k}) +
\frac{\hat{\rho}^\dagger(\mathbf{k})}{\sqrt{\omega(\mathbf{k})}}  e^{i
\omega(\mathbf{k}) t} e^{-i \mathbf{k} \cdot \mathbf{x}} a^\dagger(\mathbf{k})
\right) \psi \right].  
\end{equation} 
Then, the following equation holds:
\begin{equation} \label{eq:interactioninourmodel} \square \left\langle \psi^t,
    \varphi(t,\mathbf{x}) \psi^t \right\rangle = - \sum_{k=1}^N \left\langle
\psi^t, \rho ** \delta(\hat{\mathbf{x}}_k - \mathbf{x}) \psi^t \right\rangle,
\end{equation} 
where $\square := \partial_t^2 - \triangle_{\mathbf{x}}$, and
the double convolution defined as in \eqref{eq:doubleconvolution} is here
understood as a shorthand notation for 
\begin{equation} 
    \begin{split}
\label{eq:deltastarstar} \rho ** \delta(\hat{\mathbf{x}}_k - \mathbf{x}) & =
\int d^3\mathbf{y}_1 \int d^3\mathbf{y}_2 \ \rho(\mathbf{y}_1)
\rho(\mathbf{y}_2)
\delta(\hat{\mathbf{x}}_k-\mathbf{y}_1-(\mathbf{x}-\mathbf{y}_2)) .  \\ & =
\int d^3\mathbf{y}_1 \ \rho(\mathbf{y}_1) \rho(\mathbf{x} - \hat{\mathbf{x}}_k
+ \mathbf{y}_1).  \end{split} 
\end{equation} 
\end{theorem} %

We observe that the Ehrenfest equation \eqref{eq:interactioninourmodel} for
the scalar fiel features a ``source term'' on the right hand side. It consists
of the $N$ electrons as sources whose point-like nature is smeared out by the
form factors $\rho$ comprising the ultraviolet
cut-off. The two occurrences of $\rho$ in the double convolution $\rho **
\delta$ arise like this: In the computation, the source term is introduced by
means of the commutation relation \eqref{eq:drr}. The latter features two
occurrences of $\varphi$ whereas each $\varphi$ bares one $\rho$ in its definition in
\eqref{eq:definitionoffieldx}.\\

The remaining section of the paper provides the proofs of the above theorems.
It is divided in section~\ref{sec:theconcept}, which explains the strategy of
proof regarding existence of solutions, section~\ref{sec:proofsec1}, which
collects necessary results about the single-time evolution operators,
section~\ref{sec:proofsec2}, which constructs the multi-time evolution and
provides the proof of Theorem~\ref{thm:mainthm}, section~\ref{sec:uniqueness},
which asserts the uniqueness of solutions, i.e., Theorem~\ref{thm:uniqueness},
and finally, section~\ref{sec:interactionofdfp}, which carries out the 
computation for the proof of Theorem~\ref{thm:interaction}.

\section{Proofs} \label{sec:theproofs}

\subsection{Strategy of proof for existence of solutions}
\label{sec:theconcept}

Before treating the general case in the following sections, it is helpful to
explain our strategy of proof in the simplest case of $N=2$ as there we can
easily make the index partitions fully explicit and do not obstruct ideas in
the compact partitioning notation introduced above. For the treatment of the
general case, however, the compact notation will prove very helpful to tackle 
the additional difficulties.\\

In the case of $N=2$, we are looking for a pointwise evaluable solution $\psi:
\R^2 \to \mathscr{H}$ to the system
\begin{equation} \label{eq:syys} \begin{split}
 & \left.  \begin{array}{r}
\Big( i \partial_{t_1} \psi(t_1,t_2)\Big)(\mathbf{x}_1,\mathbf{x}_2)  = \Big( \mathcal{H}_1(t_1) \psi(t_1,t_2)\Big)(\mathbf{x}_1,\mathbf{x}_2)
\\ \Big( i \partial_{t_2} \psi(t_1,t_2)\Big)(\mathbf{x}_1,\mathbf{x}_2)  = \Big( \mathcal{H}_2(t_2) \psi(t_1,t_2)\Big)(\mathbf{x}_1,\mathbf{x}_2)
\end{array} \right\rbrace   \ \text{if} \ \| \mathbf{x}_1 - \mathbf{x}_2\| > \delta + |t_1-t_2|, 
\\ & \qquad \  \Big( i \partial_{t} \psi(t,t)\Big)(\mathbf{x}_1,\mathbf{x}_2)   \ =  \Big( \mathcal{H}_{\{1,2\}}(t) \psi(t,t) \Big)(\mathbf{x}_1,\mathbf{x}_2)    \hspace{0.36cm} \text{if} \ t_1=t_2=t, 
\end{split} \end{equation} 
where $\mathcal{H}_{\{1,2\}} = \mathcal{H}_1 + \mathcal{H}_2$. Note that there
is a little bit of redundancy in this system, since the second case is implied
by the first if $t_1=t_2$ and $ \| \mathbf{x}_1 - \mathbf{x}_2\| > \delta +
|t_1-t_2|$. The relevance of the second case comes from the points where the
times are equal, but the particles have smaller distance than $\delta$, i.e.\
the line in figure \ref{fig:Sdelta}.  

The first step is to show that evolution operators $U_{\{1\}}, U_{\{2\}},
U_{\{1,2\}}$, one for each of the single equations in \eqref{eq:syys}, exist.
These evolutions satisfy the usual properties of two-parameter propagators and,
for all $\psi$ in a suitable domain, generate a time evolution fulfilling
\begin{equation}
i \frac{\partial}{\partial t} U_A (t,s) \psi =  \mathcal{H}_A(t) U_A(t,s) \psi, \quad A \in \{ \{ 1 \}, \{ 2 \}, \{ 1,2\} \}. 
\end{equation}
An essential property of $U_A$ that we will need is that it makes the support
of a wave function grow only within its future (or
backwards) lightcone, as it is common for Dirac propagators. A further necessary
ingredient that has to be proven is the invariance of smooth functions under
the time evolutions.  This will be established by commutator theorems following
Huang \cite{huang}.

\begin{figure}[h] 
\begin{center} \includegraphics[scale=1]{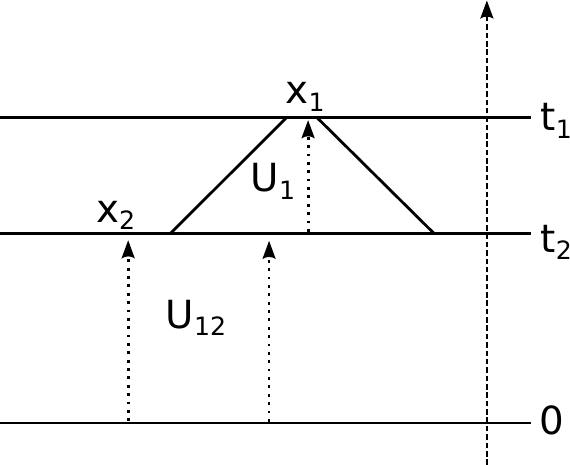}
\caption{Depiction of the multi-time evolution. First, the initial values are evolved from time $0$ to $t_2$ with the common propagator $U_{\{1,2\}}$, then only the degrees of freedom of particle $1$ are brought to time $t_1$ by applying $U_1$. This works consistently because $x_2$ is outside of the backward lightcone of $x_1$ with an additional distance of $\delta$, as sketched here.} \label{fig:Us}
\end{center} \end{figure} 

In the second step, a candidate for the solution can directly be constructed
with the help of the evolution operators $U_A$. Given smooth initial values
$\psi^0$ at $t_1=t_2=0$, we define
\begin{equation}
\psi (t_1, t_2) =  U_{\{1\}}(t_1,t_2)U_{\{1,2\}}(t_2,0)\psi^0.
\end{equation}
The idea is: First evolve both particles simultaneously up to time $t_2$ and
then only evolve the first particle to $t_1$. If more times are added, we need
to order them increasingly such that we do not ``move back and forth'' in the
time coordinates.  It is necessary, as mentioned above, to prove that the $U_A$
operators keep functions sufficiently regular to be able to define $\psi$ in a
pointwise sense and obtain a differentiable function. 

By definition, $i \partial_{t_1} \psi(t_1,t_2)  = \mathcal{H}_1(t_1)
\psi(t_1,t_2)$ holds. If both times are equal, the equation $i \partial_{t}
\psi(t,t)  =  \mathcal{H}_{\{1,2\}}(t) \psi(t,t)$ is also fulfilled. For the
derivative with respect to $t_2$, one has
\begin{equation}
\Big( i \partial_{t_2} \psi (t_1, t_2)\Big)(\mathbf{x}_1, \mathbf{x}_2) = \left( U_{\{1\}}(t_1,t_2) \mathcal{H}_2(t_2) U_{\{1,2\}}(t_2,0)\psi^0 \right) (\mathbf{x}_1, \mathbf{x}_2).
\end{equation}
To show that $\psi$ solves the multi-time equations, $U_1$ and $\mathcal{H}_2$ have to commute on the configurations with minimal space-like distance $\delta$. By taking another derivative, and after treating some difficulties that originate in the unboundedness of $\mathcal{H}_2(t_2)$, we will be able to reduce this to the consistency condition
\begin{equation}
\left( \left[ \mathcal{H}_1(t_1), \mathcal{H}_2(t_2) \right] \psi(t_1,t_2) \right) (\mathbf{x}_1,\mathbf{x}_2) = 0.
\end{equation}
The crucial ingredients in this step are that the commutators vanish at
configurations inside our domain of definition $\mathscr{S}_\delta$, and that the
supports grow at most with the speed of light.

\subsection{Dynamics of the single-time equations} \label{sec:proofsec1}

In this section, we consider a fixed set $A \subset \{ 1,...,N \}$ with the respective Hamiltonian $\mathcal{H}_A(t)$ defined in \eqref{eq:defHA} and construct a corresponding time evolution operator $U_A(t,s)$. This is contained in the following theorem, which uses the subsequent Lemmas \ref{thm:katorellichHtilde} and \ref{thm:lemmaUtilde}. The subsection continues with important properties of the operator $U_A(t,s)$, namely the spreading of data with at most the speed of light (Lemma \ref{thm:lemmacausal}) and the invariance of certain smooth functions (Lemma \ref{thm:invariance}, Corollary \ref{thm:corollaryDinvariant}), namely those in the important set $\mathscr{D}$ defined in \eqref{eq:defD}.
We denote the identity map by $\id$.

\begin{theorem} \label{thm:existsU}
There exists a unique two-parameter family of unitary operators $U_A (t,s): \mathscr{H} \to \mathscr{H}$ with the properties that for all $t,s,r \in \R$,
\begin{enumerate}
\item $U_A(t,t) = \id$,
\item $U_A(t,r) = U_A(t,s)U_A(s,r)$,
\item If $\psi \in \mathscr{D} $, then $\left. \frac{\partial }{ \partial t} U_A(t,s) \psi \right|_{t=s} = -i \mathcal{H}_A(s) \psi $.
\end{enumerate}
\end{theorem}

\begin{remark}
The third property in the theorem is slightly weaker than in the common case of
time-independent Hamiltonians, where one can prove that the derivative exists
for all functions in the domain of the Hamiltonian. But in our case, since we
do not know whether $\dom(\mathcal{H}(t))$ is independent of $t$,
we have to reside to a common domain like $\mathscr{D}$ .
\end{remark}

\begin{proof}
We first prove the \textit{existence} of $U_A$. Consider for a fixed $s \in \R$ the time-independent Hamiltonian
\begin{equation} \label{eq:defineHtildealpha}
\widetilde{\mathcal{H}}_{A,s} := \mathcal{H}_f + \sum_{j \in A}\left( \mathcal{H}^0_j + \varphi_j(s) \right).
\end{equation}
It is proven below in Lemma \ref{thm:katorellichHtilde} that this Hamiltonian is essentially self-adjoint on the dense domain $\mathscr{D}$. Therefore, there is a strongly continuous unitary one-parameter group $\widetilde{U}_{A,s}$ with the property that if $\psi \in \dom (\widetilde{\mathcal{H}}_{A,s})$, then $ \frac{\partial }{ \partial t} \widetilde{U}_{A,s}(t) \psi = -i \widetilde{\mathcal{H}}_{A,s} \psi $. We can transform back to the Hamiltonian without tilde by setting
\begin{equation}
U_A (t,s) :=  e^{ i \mathcal{H}_f (t-s)} \widetilde{U}_{A,s}(t-s) \quad \forall t,s \in \R.
\end{equation}
We have to check that the such defined two-parameter family of unitary operators satis\-fies the properties listed in the theorem. 
\begin{enumerate}
\item For all $t \in \R$, $U_A(t,t) = \id$ follows immediately by  $\widetilde{U}_{A,s}(0)=\id$.
\item We compute for any $t,s,r \in \R$,
\begin{equation} \begin{split}
U_A(t,s) U_A(s,r) & =  e^{ i \mathcal{H}_f (t-s)} \widetilde{U}_{A,s}(t-s) e^{ i \mathcal{H}_f (s-r)} \widetilde{U}_{A,r}(s-r) 
\\ & = e^{ i \mathcal{H}_f (t-r)} \underbrace{e^{ i \mathcal{H}_f (r-s)} \widetilde{U}_{A,s}(t-s) e^{ i \mathcal{H}_f (s-r)}}_{=\widetilde{U}_{A,r}(t-s) \ \text{by Lemma \ref{thm:lemmaUtilde}, part 2}} \widetilde{U}_{A,r}(s-r) 
\\ & = e^{ i \mathcal{H}_f (t-r)} \widetilde{U}_{A,r}(t-s) \widetilde{U}_{A,r}(s-r)  
\\ & = U_A(t,r).
\end{split}\end{equation}
\item Let  $\psi \in \mathscr{D}$ and $t,s \in \R$, then also $\psi \in \dom(\mathcal{H}_f) \cap \dom( \widetilde{\mathcal{H}}_{A,s})$, and 
\begin{equation} \begin{split}
& \left. i \partial_t U_A (t,s) \psi(s) \right|_{t=s} 
\\ = & \left[ -\mathcal{H}_f e^{i \mathcal{H}_f (t-s)} \widetilde{U}_{A,s}(t-s) \psi(s) + e^{i \mathcal{H}_f (t-s)} \widetilde{\mathcal{H}}_{A,s} \widetilde{U}_{A,s}(t-s) \psi(s)  \right]_{t=s} 
\\ = & \left[ -\mathcal{H}_f  \psi(t) + e^{i \mathcal{H}_f (t-s)} \widetilde{\mathcal{H}}_{A,s} e^{-i \mathcal{H}_f (t-s)} \psi(t)  \right]_{t=s} 
\\ = & \left[ -\mathcal{H}_f  \psi(t) + \widetilde{\mathcal{H}}_{A,t}  \psi(t)  \right]_{t=s} = \mathcal{H}_A (s) \psi(s),
\end{split}\end{equation}
where we used in the last line the statement of Lemma \ref{thm:lemmaUtilde}, part 1. This establishes the third property and hence existence.
\end{enumerate}
We now prove \textit{uniqueness} of $U_A$. Assume there are two families $U_A(t,s)$ and $U'_A(t,s)$ with all required properties. Pick some $\psi^0 \in  \mathscr{D}$,  then $\psi(t) := U_A(t,0) \psi^0$ and  $\psi'(t) := U'_A(t,0) \psi^0$ are differentiable w.r.t to $t$ by the invariance of $\mathscr{D}$ (Corollary  \ref{thm:corollaryDinvariant}). By li\-nea\-rity, also $w(t):=\psi(t)-\psi'(t)$ satisfies the differential equation $i \partial_t w(t) = \mathcal{H}_A(t) w(t)$. Note that $w(0) = 0$. Because $\mathcal{H}_A(t)$ is self-adjoint for all times, the norm is preserved during time evolution:
\begin{equation}
 i \partial_t \left\langle w(t), w(t) \right\rangle = - \left\langle \mathcal{H}_A(t)w(t), w(t) \right\rangle +  \left\langle w(t), \mathcal{H}_A(t)w(t) \right\rangle = 0 \end{equation}  
Therefore, also $w(t)$ must have norm zero, so $\psi(t) = \psi'(t) ~ \forall t \in \R$, which proves that the families $U_A(t,s)$ and $U'_A(t,s)$ are in fact identical. 
\qed
\end{proof}

We have used the statements of the following two lemmas:

\begin{lemma} \label{thm:katorellichHtilde}
Let $t,s \in \R$.  The Hamiltonian $\mathcal{H}_A(t)$ and the operator $\widetilde{\mathcal{H}}_{A,s}$ defined in \eqref{eq:defineHtildealpha} are essentially self-adjoint on the domain $\mathscr{D}$ defined in \eqref{eq:defD}.

\end{lemma}

The following proof is a generalization of an argument by Arai \cite{arai} and a similar argument given in \cite[app.\ C]{lieblossstability}.

\begin{proof}
Let $t,s \in \R$. We want to prove essential self-adjointness of $\widetilde{\mathcal{H}}_{A,s}$ using the commutator theorem \cite[theorem X.37]{reedsimon2}, nicely proven in \cite{farislavine}. It is easy to see that the same argumentation can then also be applied to $\mathcal{H}_A(t)$, which just has one term less. Consider
\begin{equation} \label{eq:Kcommutatortheorem}
K_A:= \sum_{j \in A} -\triangle_j + \mathcal{H}_f + 1. 
\end{equation}
This operator is essentially self-adjoint on $\mathscr{D}$ due to well-known results (see e.g.\ \cite{reedsimon2}) and certainly satisfies $K_A \geq 1$. Therefore, to apply the commutator theorem, we need to prove:
\begin{enumerate}
\item $\exists c \in \R$ such that $\forall \phi \in \mathscr{D}$, $\|(\widetilde{\mathcal{H}}_{A,s})\phi\| \leq c \| K_A \phi\|$.
\item $\exists d \in \R$ such that $\forall \phi \in \mathscr{D}$, $|\left\langle \widetilde{\mathcal{H}}_{A,s}\phi, K_\alpha \phi \right\rangle - \left\langle K_A \phi,  \widetilde{\mathcal{H}}_{A,s}\phi\right\rangle| \leq d \| K_A^{\sfrac{1}{2}} \phi\|$.
\end{enumerate}
\textit{Proof of 1.} We make use of the standard estimates (see e.g.\ \cite{nelson}) valid for all $\psi \in \dom (\mathcal{H}_f^{\sfrac{1}{2}})$ and $f \in L^2(\R^3,\C)$,
\begin{equation} \label{eq:boundsoncreationannihilation}
\left\|  \int d^3\mathbf{k} ~ f(\mathbf{k}) a(\mathbf{k}) \psi \right\| \leq \left\| f \right\|_2  \left\| \mathcal{H}_f^{\sfrac{1}{2}} \psi \right\|, \quad \left\|  \int d^3\mathbf{k} ~ f(\mathbf{k}) a^\dagger(\mathbf{k}) \psi \right\| \leq \left\| f \right\|_2  \left\|\mathcal{H}_f^{\sfrac{1}{2}} \psi \right\| + \left\| f \right\|_2 \| \psi \|.
\end{equation}
Now let $\phi \in \mathscr{D}$. We have by the triangle inequality
\begin{equation}
\left\| \widetilde{\mathcal{H}}_{A,s} \phi \right\| \leq \sum_{j \in A} \left( \left\| \mathcal{H}^0_j \phi \right\| + \left\| \varphi_j \phi \right\| \right) + \left\| \mathcal{H}_f \phi \right\|,
\end{equation}
so we need to bound each of the summands on the right hand side. $ \left\| \mathcal{H}_f \phi \right\| \leq  \left\| K_A \phi \right\|$ is clear since $1$ and $-\triangle$ are positive operators. Next we consider the free Dirac operator,
\begin{equation}
 \left\| \mathcal{H}^0_j \phi \right\| \leq m  \left\| \phi \right\| +  \left\| -i (\boldsymbol{\alpha}_j \cdot \nabla_j) \phi \right\|.
\end{equation}
The derivative term needs closer inspection,
\begin{equation}
 \left\| -i (\boldsymbol{\alpha}_j \cdot \nabla_j) \phi \right\|^2 = \left\langle \phi, -(\boldsymbol{\alpha}_j \cdot \nabla_j)^2 \phi \right\rangle = \left\langle \phi, - \triangle \phi \right\rangle,
\end{equation}
where only the Laplacian survives because the $\alpha$-matrices anticommute and the derivatives commute. Continuing with the Cauchy-Schwarz inequality and the elementary inequality $\sqrt{ab} \leq \frac{1}{2} (a + b) \ \forall a,b \geq 0$, we obtain
\begin{equation}
 \left\| -i (\boldsymbol{\alpha}_j \cdot \nabla_j) \phi \right\| \leq \sqrt{\left\langle \phi, - \triangle \phi \right\rangle} \leq \sqrt{ \left\| \phi \right\|  \left\| - \triangle \phi \right\|} \leq \frac{1}{2} \left(  \left\| \phi \right\| +  \left\| -\triangle \phi \right\| \right).
\end{equation}
Again, since all the summands in $K_A$ are positive operators, this directly leads to
\begin{equation}
 \left\| \mathcal{H}^0_j \phi \right\| \leq C  \left\| K_A \phi \right\|.
\end{equation}
In the whole article, $C$ denotes an arbitrary positive constant that may be different each time. For the interaction term, we see that the factor $\frac{\hat{\rho}(\mathbf{k})}{\sqrt{\omega(\mathbf{k})}}$ is in $L^2$ since $\hat{\rho}$ being a Schwartz function ensures rapid decay at infinity and since the singularity at $k=0$ (present only for $\mu=0$)  is integrable. This allows the use of \eqref{eq:boundsoncreationannihilation}, giving
\begin{equation} \label{eq:boundonphi}
 \left\| \varphi_j \phi \right\| \leq C \left( \left\| \mathcal{H}_f^{\sfrac{1}{2}} \phi \right\| + \left\| \phi \right\| \right),
\end{equation}
and with one more application of Cauchy-Schwarz,
\begin{equation}
\left\| \mathcal{H}_f^{\sfrac{1}{2}} \phi \right\| \leq \left\| \phi \right\|^{\sfrac{1}{2}}  \left\| \mathcal{H}_f \phi \right\|^{\sfrac{1}{2}} \leq \frac{1}{2} \left(\left\| \phi \right\| + \left\| \mathcal{H}_f \phi \right\| \right), 
\end{equation}
we are done with the proof that there is a constant $c$ (not depending on $\phi$) with  $\|(\widetilde{\mathcal{H}}_{A,s})\phi\| \leq c \| K_A \phi\|$.
\\ \textit{Proof of 2.} As in the previous step, we can bound the summands in $\widetilde{\mathcal{H}}_{A,s}$ one by one. We first observe that $\mathcal{H}_f$ and $\mathcal{H}^0_j$ commute with $K_A$. For the interaction term, we have
\begin{equation} \label{eq:Sabine}
\left[ \varphi_j, K_A \right] = \left[ \varphi_j, - \triangle_j \right] + \left[ \varphi_j, \mathcal{H}_f \right],
\end{equation}
so let us compute
\begin{equation} \label{eq:calculationphiderivative} \begin{split}
& \left\langle \varphi_j \phi, \triangle_j \phi \right\rangle - \left\langle \triangle_j \phi, \varphi_j \phi \right\rangle =  \sum_{a=1}^3 \left\langle \frac{\partial}{\partial x_j^a}  \phi,  \frac{\partial}{\partial x_j^a} \varphi_j \phi \right\rangle - \left\langle  \frac{\partial}{\partial x_j^a} \varphi_j \phi, \frac{\partial}{\partial x_j^a} \phi \right\rangle  
\\ & = ~  2i \sum_{a=1}^3 \Im \left\langle \frac{\partial}{\partial x_j^a}  \phi,  \frac{\partial}{\partial x_j^a} \varphi_j \phi \right\rangle 
\\ & = ~  2i \sum_{a=1}^3 \Im \left\langle \frac{\partial}{\partial x_j^a}  \phi, \int d^3\mathbf{k} ~ \left[ \left( \frac{-ik^a_j \hat{\rho}(\mathbf{k}) }{\sqrt{\omega(\mathbf{k})}} e^{-i \omega(\mathbf{k}) t} e^{i \mathbf{k} \cdot \hat{\mathbf{x}}_j}  a(\mathbf{k}) + c.c.\ \right) \phi \right]  \right\rangle, 
\end{split}\end{equation}
where the last equality holds since $\left\langle \frac{\partial}{\partial x_j^a}  \phi, \varphi_j \frac{\partial}{\partial x_j^a} \phi \right\rangle$ is real; and ``c.c'' denotes the hermitian conjugate of the preceding term.  Since $\hat{\rho}$ is a Schwartz function, also $-ik_j^a \hat{\rho}(\mathbf{k})$ is, so we get from the estimate \eqref{eq:boundsoncreationannihilation}
\begin{equation}
\left| \left\langle \varphi_j \phi, \triangle_j \phi \right\rangle - \left\langle \triangle_j \phi, \varphi_j \phi \right\rangle \right| \leq C \left( \left\| \mathcal{H}_f^{\sfrac{1}{2}} \phi \right\| + \left\| \phi \right\| \right)  \leq 2C \left\| K_A^{\sfrac{1}{2}} \phi \right\|.
\end{equation}
For the second term in \eqref{eq:Sabine}, we look at the commutator of $\varphi_j$ and $\mathcal{H}_f$. This amounts to a time derivative of $\varphi_j(t)$, which gives an expression like in the last line of \eqref{eq:calculationphiderivative}, but where the function $-ik^a_j \hat{\rho}(\mathbf{k})$ is replaced by $-i \omega(\mathbf{k}) \hat{\rho}(\mathbf{k})$. This is again a Schwartz function. Using estimate \eqref{eq:boundsoncreationannihilation} again for that function, we obtain
\begin{equation}
|\left\langle \varphi_j \phi, \mathcal{H}_f \phi \right\rangle - \left\langle \mathcal{H}_f \phi, \varphi_j \phi \right\rangle| \leq C  \left\| K_A^{\sfrac{1}{2}} \phi \right\|.
\end{equation}
This means we have shown that there is a constant $d$ (independent of $\phi$), such that \begin{equation}
\left| \left\langle \widetilde{\mathcal{H}}_{A,s}\phi, K_A \phi \right\rangle - \left\langle K_A \phi,  \widetilde{\mathcal{H}}_{A,s}\phi\right\rangle \right| \leq d \| K_A^{\sfrac{1}{2}} \phi\|.
\end{equation}
This is the second necessary ingredient for the application of the commutator theorem, which gives the statement of the lemma. 
\qed
\end{proof}

\begin{lemma} \label{thm:lemmaUtilde}
The self-adjoint Hamiltonian $\widetilde{\mathcal{H}}_{A,s}$  and the unitary group $\widetilde{U}_{A,s}$ it generates satisfy the following properties for all $r,s,t \in \R$:
\begin{enumerate}
\item $e^{i \mathcal{H}_f (t-s)} \widetilde{\mathcal{H}}^n_{A,s} e^{-i \mathcal{H}_f (t-s)} = \widetilde{\mathcal{H}}^n_{A,t} \quad \forall n \in \N$, whenever both sides are well-defined.
\item $e^{ i \mathcal{H}_f (r-s)} \widetilde{U}_{A,s}(t-s) e^{ i \mathcal{H}_f (s-r)} = \widetilde{U}_{A,r}(t-s)$.
\end{enumerate}
\end{lemma}

\begin{proof} Let $r,s,t \in \R$.
\begin{enumerate}
\item We have for $n=1$
\begin{equation}
\begin{split}
e^{i \mathcal{H}_f (t-s)} \widetilde{\mathcal{H}}_{A,s} e^{-i \mathcal{H}_f (t-s)} & = e^{i \mathcal{H}_f (t-s)}\left( \mathcal{H}_f + \sum_{j \in A} \mathcal{H}^0_j +  \varphi_j(s) \right) e^{-i \mathcal{H}_f (t-s)} 
\\ & = \mathcal{H}_f + \sum_{j \in A} \mathcal{H}^0_j + e^{i \mathcal{H}_f (t-s)}  \varphi_j(s) e^{-i \mathcal{H}_f (t-s)} 
\\ & =  \widetilde{\mathcal{H}}_{A,t}.
\end{split}\end{equation}
The statement for arbitrary $n \in \N$ follows directly from the $n=1$ case, which can be seen by inserting the identity $\id = e^{-i \mathcal{H}_f (t-s)} e^{i \mathcal{H}_f (t-s)}$ between the factors of  $\widetilde{\mathcal{H}}_{A,s}$,
\begin{equation}
e^{i \mathcal{H}_f (t-s)} \widetilde{\mathcal{H}}^n_{A,s} e^{-i \mathcal{H}_f (t-s)} =  \prod_{k=1}^n e^{i \mathcal{H}_f (t-s)} \widetilde{\mathcal{H}}_{A,s} e^{-i \mathcal{H}_f (t-s)} = \widetilde{\mathcal{H}}^n_{A,t}.
\end{equation}
\item By the analytic vector theorem, the set $\mathscr{A}$ of analytic vectors for $\widetilde{\mathcal{H}}_{A,s}$ is dense. Hence its image under the unitary map $e^{i \mathcal{H}_f (r-s)}$ is also dense. Let $\psi \in e^{i \mathcal{H}_f (r-s)}\left( \mathscr{A} \right)$. We can write 
\begin{align}
e^{ i \mathcal{H}_f (r-s)} \widetilde{U}_{A,s}(t-s) e^{ i \mathcal{H}_f (s-r)}\psi & = e^{ i \mathcal{H}_f (r-s)} \sum_{n=0}^\infty \frac{i^n (t-s)^n}{n!} \widetilde{\mathcal{H}}_{A,s}^n e^{ i \mathcal{H}_f (s-r)} \psi \nonumber
\\ & = \sum_{n=0}^\infty \frac{i^n (t-s)^n}{n!}e^{ i \mathcal{H}_f (r-s)} \widetilde{\mathcal{H}}_{A,s}^n e^{ i \mathcal{H}_f (s-r)} \psi  \nonumber
\\ & = \sum_{n=0}^\infty \frac{i^n (t-s)^n}{n!} \widetilde{\mathcal{H}}_{A,r}^n  \psi , 
\end{align}
where we used part 1 of the lemma in the last step. The series converges, so $\psi$ is analytic for $\widetilde{\mathcal{H}}_{A,r}$, which proves
\begin{equation} \label{eq:eqindenseset}
e^{ i \mathcal{H}_f (r-s)} \widetilde{U}_{A,s}(t-s) e^{ i \mathcal{H}_f (s-r)}\psi  = \widetilde{U}_{A,r}(t-s) \psi, \quad \forall  \psi \in e^{i \mathcal{H}_f (r-s)}\left( \mathscr{A} \right).
\end{equation}
Equation \eqref{eq:eqindenseset} tells us that the bounded operators $e^{ i \mathcal{H}_f (r-s)} \widetilde{U}_{A,s}(t-s) e^{ i \mathcal{H}_f (s-r)}$ and $\widetilde{U}_{A,r}(t-s)$ agree on a dense set, which implies that they are equal. \qed
\end{enumerate} 
\end{proof}

The next lemma is about the causal structure of our equations. It uses the usual definition of addition of sets,
\begin{equation}
M + R := \{ m+r | m \in M, r \in R \}.
\end{equation}
In order to simplify notation, it is implied that vectors in $\R^{3N}$ and $\R^3$ can be added by just changing the respective $j$-th coordinate, e.g.\ $(\mathbf{x}_1, \mathbf{x}_2) + \mathbf{y}_2 \equiv (\mathbf{x}_1, \mathbf{x}_2 + \mathbf{y}_2)$.

\begin{lemma} \label{thm:lemmacausal}\noindent
\begin{enumerate}
\item The evolution operators $U_A$ do not propagate data faster than light, i.e.\ if for $R \subset \R^{3N}$ we have $\supp \psi \subset R$, then for all $t,s \in \R$,
\begin{equation} \label{eq:harry}
\supp \left(U_A(t,s)\psi \right) \subset R + \sum_{j \in A} B_{|t-s|}(\mathbf{x}_j).
\end{equation} 
\item Let $\psi$ be the solution of $i \partial_{t} \psi = \mathcal{H}_A(t) \psi$ with smooth initial values given as \\ $\psi(0,...,0) = \psi^0$. Then for all $t \in \R$, $\psi(t, \mathbf{x}_1, ..., t, \mathbf{x}_N)$ is uniquely determined by specifying initial conditions on $\sum_{j \in A} \overline{B}_{|t|}(\mathbf{x}_j)$.
\end{enumerate}
\end{lemma}

\begin{proof} \noindent
\begin{enumerate}
\item This lightcone property of the free Dirac equation is well-known (compare \cite[theorem 2.20]{Deckert2014}). The claim for our model is a direct generalization to the many-particle case of the functional analytic arguments in \cite[theorem 3.4]{zenkstockmeyer}. (Note that it is also feasible to adapt the arguments using current conservation in \cite[lemma 14]{multitimeNoPotentials} since the continuity equation holds for our model, as well.)
\item This follows directly from 1. since if $\psi$ and ${\psi}'$ are two solutions whose initial values $\psi^0$ and $\psi'^0$ agree on $\sum_{j \in A} \overline{B}_{|t|}(\mathbf{x}_j)$, then
\begin{equation}
\supp (\psi^0 - \psi'^0) \subset \R^{3N} \setminus \sum_{j \in A} \overline{B}_{|t|}(\mathbf{x}_j)
\end{equation}
implies by \eqref{eq:harry}
\begin{equation}
\supp (U_A(t,0) (\psi^0 -  \psi'^0)) \subset \R^{3N} \setminus \sum_{j \in A} \overline{B}_{|t|}(\mathbf{x}_j) + \sum_{j \in A} {B}_{|t|}(\mathbf{x}_j) =  \R^{3N} \setminus \{ (\mathbf{x}_1,...,\mathbf{x}_N) \},
\end{equation}
which is the claim. \qed
\end{enumerate}
\end{proof}

Another necessary information is which domains stay invariant under the time evolutions we have just constructed. The idea is to exploit a theorem by Huang \cite[thm.\ 2.3]{huang}, which we cite here adopted to our notation.

\begin{theorem} \textbf{(Huang).} \label{thm:huang}
Let $K$ be a positive self-adjoint operator and define $Z_j(t) = K^{j-1} \left[ \mathcal{H}_{A}(t),K \right] K^{-j}$. Suppose that $Z_k(t)$ is bounded with $ \| Z_k( \cdot ) \| \in L^1_{loc}(\R)$ for all $k \leq j$. Then $U_A(t,s) [ \dom(K^j) ] =  \dom(K^j)$.
\end{theorem}

We will use a family of comparison operators for $j \in \N$, abbreviating $\sum_{k=1}^N -\triangle_k =: -\triangle$,
\begin{equation}
K_n:=  (-\triangle)^n + (\mathcal{H}_f)^n + 1.
\end{equation} 
The operator $K_n$ resembles the $n-th$ power of the operator $K_A$  we defined in \eqref{eq:Kcommutatortheorem} for the commutator theorem. Its domain of self-adjointness is denoted by $\dom(K_n)$. 

\begin{lemma} \label{thm:invariance}
The family of operators $U_A(t,s)$ with $t,s \in \R$ leaves the set $\dom(K_n)$ invariant for all $n \in \N$.
\end{lemma} 

\begin{proof}
Let $n \in \N$. It is known that $K_n$ is self-adjoint and strictly positive. We prove the invariance of $\dom(K^n)$ using Thm.\ \ref{thm:huang}, hence we only need the case $j=1$   
and need to bound $Z_1(t) = [ \mathcal{H}_A(t), K_n ] K_n^{-1}$. 
\\ Note that, since $K_n$ is positive, $0$ is in its resolvent set. This means that $K_n: \dom(K_n) \to \mathscr{H}$ is bijective, so its inverse $K_n^{-1}: \mathscr{H} \to \dom(K_n)$ is bounded by the closed graph theorem. Because the Laplacian commutes with the free Dirac operator (in the sense of self-adjoint operators, which can e.g.\ be seen by their resolvents), this carries over to $(-\triangle)^n$ and the commutator gives
\begin{equation}
\left[ \mathcal{H}_A(t), K_n \right] = \sum_{j \in A}  \left[ \varphi_j(t), (-\triangle)^n \right] + \sum_{j \in A} \left[ \varphi_j(t), \mathcal{H}^n_f \right] .
\end{equation} 
The commutator terms give rise to derivatives of the field terms $\varphi$, similarly as in the calculation \eqref{eq:calculationphiderivative}. It becomes apparent that arbitrary derivatives with respect to time or space variables lead to the multiplication of $\hat{\rho}(\mathbf{k})$  in \eqref{eq:definitionoffield} by a product of $k^a$ and $\omega(\mathbf{k})$ factors, which still keep the rapid decay at infinity. Therefore, also the derivative is a quantum field with an $L^2$-function as cut-off function. This means that the bound \eqref{eq:boundonphi} can analogously be applied to the commutator and we have some $C > 0$ with
\begin{equation}
\left\| \left[ \mathcal{H}_A(t), K_n \right] \eta \right\| \leq C \left( \| \mathcal{H}_f \eta \| + \| \eta \| \right) \quad \forall \eta \in \dom(K).
\end{equation}
By the inequality of arithmetic and geometric mean, 
\begin{equation}
\| \mathcal{H}_f \eta \| = \| \sqrt[n]{\mathcal{H}_f^n} \eta \| \leq C ( \| \mathcal{H}_f^n \eta \| + \| \eta \|) \leq  C ( \| K_n \eta \| + \| \eta \|)
\end{equation} 

Since $K_n^{-1} \psi \in \dom(K_n)$, we can apply this to $Z_1(t)$,
\begin{equation}
\left\| Z_1(t) \psi \right\| = \left\| \left( \left[ \mathcal{H}_A(t), K_n \right] \right) K_n^{-1} \psi \right\| \leq C \left( \| K_n K_n^{-1} \psi \| + \| K_n^{-1} \psi \| \right) = C \left(1 +  \|K_n^{-1}\|_{op}   \right) |\psi\|,
\end{equation}
which implies that $Z_1(t)$ is bounded with $\| Z_1(\cdot)\| \in L^1_{loc}(\R)$. Hence, application of Theorem \ref{thm:huang} yields the claim. 
\qed
\end{proof}

\begin{corollary} \label{thm:corollaryDinvariant}
The family of operators $U_A(t,s)$ with $t,s \in \R$ leaves the set $\mathscr{D}$, defined in \eqref{eq:defD}, invariant.
\end{corollary} 

\begin{proof}
By Lemma \ref{thm:invariance}, $U_A(t,s)$ with $t,s \in \R$ leaves $\dom(K_n)$ invariant for each $n \in \N$. We claim that 
\begin{equation} \label{eq:67}
\dom(K_n) = \dom((-\triangle)^n) \otimes \F \cap L^2(\R^{3N},\C^K) \otimes \dom(\mathcal{H}_f^n).
\end{equation}
The operator $K_n$ is of the form $(-\triangle)^n \otimes \id + \id \otimes \mathcal{H}^n_f + 1$, where the bounded operator $1$ is irrelevant for the domain. By \cite[chap.\ VIII.10]{reedsimon1}, an operator of this structure on a tensor product space is essentially self-adjoint on the domain $\dom((-\triangle)^n) \otimes \F \cap L^2(\R^{3N},\C^K) \otimes \dom(\mathcal{H}_f^n)$. The domain of self-adjointness arises when we take the closure of that operator. It is, however, known from \cite[p.\ 160]{schmuedgen} that a sum of positive operators is already closed on the domain \eqref{eq:67}. Thus, \eqref{eq:67} is actually the domain of self-adjointness of $K_n$.  

Let $\psi \in \mathscr{D}$, then also $\psi \in \dom(K_n)$ for all $n \in \N$. Thus, $U_A(t,s)\psi \in \dom(K_n)$ for all $n \in \N$. For the Fock space part, this directly gives
\begin{equation}
U_A(t,s) \psi \in \bigcap_{n = 1}^\infty L^2(\R^{3N},\C^K) \otimes \dom(\mathcal{H}_f^n) = L^2(\R^{3N},\C^K) \otimes \dom(\mathcal{H}_f^\infty).
\end{equation}
In the $L^2$-part, we first note that Lemma \ref{thm:lemmacausal} gives an upper bound on the growth of supports, so compactness of the support is preserved under the time evolution $U_A(t,s)$. Secondly, we have 
\begin{equation}
U_A(t,s) \psi \in \bigcap_{n = 1}^\infty \dom((-\triangle)^n) \otimes \F \subset C^\infty(\R^{3N},\C^K) \otimes \F,
\end{equation}
which follows from Sobolev's lemma as contained in the proposition in \cite[chap.\ IX.7]{reedsimon2}. These two facts imply that the time evolution leaves $C_c^\infty$ invariant. Thus we infer $U_A(t,s) \psi \in \mathscr{D}$. \qed
\end{proof}

Another result that will be helpful later is that not only the time evolutions leave the set $\mathscr{D}$ invariant, but also the terms in the Hamilton operators themselves.

\begin{lemma} \label{thm:invariantdomainforops}
The set $\mathscr{D}$ is left invariant by $\mathcal{H}_f$, $\mathcal{H}^0_j$ and $\varphi_j(t)$ for each $1 \leq j \leq N$ and $t \in \R$.
\end{lemma}

\begin{proof} \begin{enumerate}
\item $\mathcal{H}^0_j$ only acts on the first tensor component and on that one, it leaves $C^\infty_c$-functions invariant because it is a linear combination of partial derivatives and the identity. 
\item $\mathcal{H}_f$ only acts on the second tensor component and on that one, it leaves $\dom(\mathcal{H}_f^\infty)$ invariant by definition.
\item First we note that $\varphi_j$ does not increase supports. Now let $k \in \N, t \in \R$ and $\psi \in \dom(\mathcal{H}_f^{k+1})$. Then, using the same estimates as in the proof of Lemma \ref{thm:katorellichHtilde},
\begin{equation}
\left\| \mathcal{H}_f^{k} ~ \varphi_j(t) \psi \right\| \leq \left\|  \varphi_j(t) \mathcal{H}_f^{k} \psi \right\| + \left\|  \frac{\partial^{k}}{\partial t^{k}} \varphi_j(t) \psi \right\| \leq C \left(  \left\|  \mathcal{H}_f^{k+1} \psi \right\| + \left\|  \mathcal{H}_f \psi \right\| + \left\|  \psi \right\| \right) < \infty,
\end{equation}
which shows that $\varphi_j(t) \psi \in \dom(\mathcal{H}_f^{k})$ for every $t \in \R$. An analogous argument can be done for the operators $\mathcal{H}^0_j$, which together implies that $\varphi_j(t)$ leaves $\mathscr{D}$ invariant.  \qed
\end{enumerate} 
\end{proof}

\subsection{Construction of the multi-time evolution} \label{sec:proofsec2}

The construction of the solution of our multi-time system \eqref{eq:systemstrongsolution} relies on the consistency condition which we prove now.

\begin{lemma} \label{thm:lemmacommutators}
Let $\psi \in \mathscr{D}$ and $A, B$ be disjoint subsets of $\{ 1,...,N\}$, then the consistency condition
\begin{equation} 
\left[ \mathcal{H}_A (t_A), \mathcal{H}_B(t_B) \right] \psi (\mathbf{x}_1,...,\mathbf{x}_N)=0
\end{equation}
is satisfied whenever $\forall j \in A, k \in B: \|\mathbf{x}_j - \mathbf{x}_k\| > \delta + |t_A-t_B|$.
\end{lemma}

\begin{proof}
Let $t_A, t_B \in \R$. The commutator reads 
\begin{equation}
\left[ \mathcal{H}_A (t_A), \mathcal{H}_B(t_B) \right] = \left[ \sum_{j \in A} \mathcal{H}^0_j + \varphi_j(t_A),  \sum_{k \in B} \mathcal{H}^0_k + \varphi_k(t_B) \right] = \sum_{j \in A, k \in B} \left[ \varphi_j(t_A), \varphi_k(t_B) \right],
\end{equation}
since, by definition, the free Dirac Hamiltonians commute with the other terms. We will now show that each of the summands in the double sum applied to $\psi \in \mathscr{D}$ vanishes when evaluated at $(\mathbf{x}_1,...,\mathbf{x}_N) \in \R^{3N}$ with $\forall j \in A, k \in B: \|\mathbf{x}_j - \mathbf{x}_k\| > \delta + |t_A-t_B|$.
\\ It is well-known (e.g.\ \cite[thm\ X.41]{reedsimon2}) that field operators as defined in \eqref{eq:definitionoffield} satisfy the CCR, which means
\begin{equation}\label{eq:comparwiththis} \begin{split}
& \left[  \varphi_j(t_A),  \varphi_k(t_B)  \right] \psi(\mathbf{x}_1,...,\mathbf{x}_N) 
\\ = ~ &   i \Im \int \frac{d^3\mathbf{k}}{\omega(\mathbf{k})} ~  \hat{\rho}(\mathbf{k})^\dagger \hat{\rho}(\mathbf{k})  e^{i \omega(\mathbf{k}) t_A -i \mathbf{k} \cdot {\mathbf{x}}_j} e^{-i \omega(\mathbf{k}) t_B + i \mathbf{k} \cdot {\mathbf{x}}_k} \psi(\mathbf{x}_1,...,\mathbf{x}_N)  
 \\ = ~ & \frac{1}{2} \int \frac{d^3\mathbf{k}}{\omega(\mathbf{k})}  \left(  \hat{\rho}(\mathbf{k})^\dagger \hat{\rho}(\mathbf{k}) e^{i \omega(\mathbf{k}) (t_A - t_B)} e^{-i\mathbf{k} \cdot (\mathbf{x}_j - \mathbf{x}_k)} - \ \text{c.c} \ \right) \psi(\mathbf{x}_1,...,\mathbf{x}_N)  
 \\ = ~ & \frac{1}{2} \int \frac{d^3\mathbf{k}}{\omega(\mathbf{k})}  \left( \int d^3\mathbf{y}_1 d^3\mathbf{y}_2  \rho(\mathbf{y}_1) e^{-i\mathbf{k} \cdot \mathbf{y}_1} \rho(\mathbf{y}_2) e^{i\mathbf{k} \cdot \mathbf{y}_2} e^{i \omega(\mathbf{k}) (t_A - t_B)} e^{-i\mathbf{k} \cdot (\mathbf{x}_j - \mathbf{x}_k)} \right.  \\ & \hspace{1.7cm} \ - \text{c.c} \ \bigg)  \psi(\mathbf{x}_1,...,\mathbf{x}_N) , 
\end{split}\end{equation}
upon insertion of the Fourier transforms. We compare this to the so-called Pauli-Jordan function \cite[p.\ 88]{scharf}, i.e. the distribution
\begin{equation} \label{eq:Paulijordan}
\Delta(x_j,x_k) := c \int \frac{d^3\mathbf{k}}{\omega(\mathbf{k})} \left( e^{i \omega(\mathbf{k}) (t_j-t_k) - i \mathbf{k} \cdot (\mathbf{x}_j - \mathbf{x}_k) } - \ \text{c.c.} \ \right),
\end{equation}
where $c = \frac{i}{16\pi^3}$. It is known that $\Delta(x_1,x_2)=0$ whenever $x_1$ is space-like to $x_2$ \cite[p.\ 89]{scharf}. We define a double convolution by
\begin{equation} \label{eq:doubleconvolution} \begin{split}
(\rho * * \Delta) (t_j, \mathbf{x}_j, t_k, \mathbf{x}_k) := \int d^3\mathbf{y}_1 d^3\mathbf{y}_2 ~ \rho(\mathbf{y}_1) \rho(\mathbf{y}_2) \Delta(t_j, \mathbf{x}_j - \mathbf{y}_1, t_k, \mathbf{x}_k - \mathbf{y}_2)
\\ =  c \int \frac{d^3\mathbf{k}}{\omega(\mathbf{k})} \int  d^3\mathbf{y}_1 d^3\mathbf{y}_2 \rho(\mathbf{y}_1) \rho(\mathbf{y}_2)  \left( e^{i \omega(\mathbf{k}) (t_j-t_k) - i \mathbf{k} \cdot (\mathbf{x}_j - \mathbf{y}_1 - \mathbf{x}_k + \mathbf{y}_2) } - \ \text{c.c.} \ \right)
,
\end{split} \end{equation}
which is a well-defined integral since $\rho \in C^\infty_c(\R^3)$. Comparison to \eqref{eq:comparwiththis} yields
\begin{equation} \label{eq:commutatorofphis}
\frac{2}{c} \left[  \varphi_j(t_A),  \varphi_k(t_B)  \right] \psi(\mathbf{x}_1,...,\mathbf{x}_N)  = (\rho * * \Delta) (t_A, \mathbf{x}_j, t_B, \mathbf{x}_k) \psi(\mathbf{x}_1,...,\mathbf{x}_N). 
\end{equation}
We know that  $\|\mathbf{x}_j - \mathbf{x}_k\| > |t_A-t_B| + \delta$ and by \eqref{eq:cutoff}, $\rho(\mathbf{y}) \neq 0$ only  if $\| \mathbf{y} \| < \frac{\delta}{2}$. Thus the argument of the function $\Delta$ in the double convolution \eqref{eq:doubleconvolution} satisfies 
\begin{equation}\label{eq:spacelikedeltavanishes} \begin{split}
\| \mathbf{x}_j -\mathbf{y}_1 - ( \mathbf{x}_k -  \mathbf{y}_2) \| & \geq \| \mathbf{x}_k -  \mathbf{x}_j \| - \| \mathbf{y}_1 \| - \| \mathbf{y}_2 \| \\ & \geq \| \mathbf{x}_j - \mathbf{x}_k \| -  \delta \\ & > |t_A-t_B|,  \end{split}
\end{equation} 
i.e.\ it is space-like, which implies that $(\rho * * \Delta) (t_A, \mathbf{x}_j, t_B, \mathbf{x}_k) = 0$ and hence also the commutator is zero. \qed
\end{proof}
With all the previous results at hand, the existence of solutions can be treated constructively. We first prove a lemma which contains the crucial ingredient for the subsequent theorem.

\begin{lemma} \label{thm:subclaim}
Let $\zeta \in \mathscr{D}$. Let $A, B$ be arbitrary subsets of $\{1,...,N\}$ with $A \cap B = \emptyset$, let $t_B \geq s \geq t_A$, then 
\begin{equation} \label{eq:thiscommutatorisit}
\left( \left[  \mathcal{H}_A(t_A),U_B(t_B,s) \right] \zeta \right) (\mathbf{x}_1,...,\mathbf{x}_N) = 0.
\end{equation}  
holds at every point $(\mathbf{x}_1,...,\mathbf{x}_N) \in \R^{3N}$ for which $\forall j \in A, k \in B$, $\| \mathbf{x}_j - \mathbf{x}_k\| > \delta + t_B - t_A$.
\end{lemma}

The idea of the proof is to take the derivative of the commutator in \eqref{eq:thiscommutatorisit} with respect to $t_B$ to get an expression where the consistency condition proven in Lemma \ref{thm:lemmacommutators} becomes useful. However, it is not immediately clear if a term of the form $\mathcal{H}_A(t_A) U_B(t_B,s)$ is differentiable or even continuous in $t_B$ because $\mathcal{H}_A$ is not a continuous operator. Therefore, we have to take a detour and approximate $\mathcal{H}_A$ by bounded operators. A similar approximation by bounded operators is used in the proof of the Hille-Yosida theorem in \cite[ch.\ 7.4]{evans}.

\begin{proof}
Let $A,B \subset \{1,...,N\}$ with $A \cap B = \emptyset$, $s, t_A, t_B \in \R$ with $t_B \geq s \geq t_A$, $\zeta \in \mathscr{D}$ and $(\mathbf{x}_1,...,\mathbf{x}_N) \in \R^{3N}$ such that $\forall j \in A, k \in B$: $\| \mathbf{x}_j - \mathbf{x}_k\| > \delta + t_B - t_A$. 

We abbreviate $\sum_{k \in A} \varphi_k(t) =: \varphi_A(t)$ for $t \in \R$. First note that the free Dirac terms in $\mathcal{H}_A$ trivially commute, so 
\begin{equation}
\left( \left[ U_B(t_B,s), \mathcal{H}_A(t_A) \right] \zeta \right) (\mathbf{x}_1,...,\mathbf{x}_N)  = \left( \left[ U_B(t_B,s), \varphi_A(t_A) \right] \zeta \right) (\mathbf{x}_1,...,\mathbf{x}_N).
\end{equation}
Now define for $\varepsilon>0, t \in \R$ a family of auxiliary operators 
\begin{equation}
\varphi^\varepsilon_A(t):= \frac{\varphi_A(t)}{1+i\varepsilon \varphi_A(t)},
\end{equation}
which are well-defined since $\varphi_A(t)$ is self-adjoint for all $t$ \cite{spohn}. For $\lambda \in \R, \varepsilon > 0$, 
\begin{equation}
\left| \frac{\lambda}{1 + i \varepsilon \lambda} \right| \leq \frac{1}{\varepsilon} \quad \Longrightarrow \|\varphi^\varepsilon_A(t)\| \leq \frac{1}{\varepsilon}
\end{equation}
where the implication follows by the spectral theorem.
The difference of field operator $\varphi_A$ and its approximation  $\varphi_A^\varepsilon$ can be recast into
\begin{equation}
( \varphi_A(t_A) - \varphi_A^\varepsilon(t_A)) =  \frac{\varphi_A(t_A) + i\varepsilon \varphi_A(t_A)^2}{1 + i \varepsilon \varphi_A(t_A)} - \frac{\varphi_A(t_A)}{1 + i \varepsilon \varphi_A(t_A)} =   \frac{i \varepsilon}{1 + i \varepsilon \varphi_A(t_A)} \varphi_A(t_A)^2
\end{equation}
and we note the bound for all $\varepsilon > 0$:
\begin{equation}
\left\| \frac{1}{1 + i \varepsilon \varphi_A(t_A)} \right\| \leq 1.
\end{equation}
Because $U_B(t_B,s)\zeta \in \mathscr{D}$ by corollary \ref{thm:corollaryDinvariant}, we find the bound
 \begin{equation} \begin{split}
\left\| \left[ U_B(t_B,s), \varphi_A(t_A) - \varphi_A^\varepsilon(t_A) \right] \zeta \right\| & \leq \ \left\| (\varphi_A(t_A) - \varphi_A^\varepsilon(t_A)) \zeta \right\| \\ & + \ \left\| (\varphi_A(t_A) - \varphi_A^\varepsilon(t_A)) U_B(t_B,s) \zeta \right\|
\\ & \leq \ \varepsilon \left( \| \varphi_A(t_A)^2 \zeta \| + \|  \varphi_A(t_A)^2 U_B(t_B,s) \zeta \| \right).
 \end{split}
 \end{equation}
Since we can take $\varepsilon \to 0$, the norm of the left hand side has to vanish. Because we furthermore know that $\left[ U_B(t_B,s), \varphi_A(t_A) - \varphi_A^\varepsilon(t_A) \right]$ is a continuous function, the following implication holds: 
\begin{equation} \label{eq:somethingfollows}
\left( \left[ U_B(t_B,s), \varphi_A^\varepsilon(t_A) \right] \zeta \right) (\mathbf{x}_1,...,\mathbf{x}_N) = 0 \ \forall \varepsilon>0 \Rightarrow \left( \left[ U_B(t_B,s), \varphi_A(t_A) \right] \zeta \right) (\mathbf{x}_1,...,\mathbf{x}_N) = 0.
\end{equation}
Thus it remains to prove that the commutator defined for $t \in \R$,
\begin{equation}
\Omega_{t}:= \left[ U_B(t,s), \varphi^\varepsilon_A(t_A) \right] \zeta,
\end{equation}
vanishes at $(\mathbf{x}_1,...,\mathbf{x}_N)$. Note that $\Omega_t$ depends on $\varepsilon$, which we do not write for brevity. As a merit of our approximation, $t \mapsto \Omega_t$ is a continuous map $\R \to \mathscr{H}$. We proceed in four steps:
\begin{enumerate}
\item Construct an auxiliary function $\phi_t$ that solves for $\eta \in \mathscr{D}$
\begin{equation}
i \partial_t \left\langle \eta, \phi_t \right\rangle  = \left\langle \eta, \left[ \varphi_B(t), \varphi_A^\varepsilon(t_A) \right] U_B(t,s) \eta \right\rangle + \left\langle \mathcal{H}_B(t) \eta,  \phi_t  \right\rangle.
\end{equation}
\item Show that $\forall \eta \in \mathscr{D}: i \partial_t \left\langle \eta, \phi_t - \Omega_{t} \right\rangle = \left\langle \mathcal{H}_B(t) \eta, \phi_t - \Omega_{t} \right\rangle $.
\item Show that the weak equation proven in step 2 has a unique solution, thus $\phi_t = \Omega_t$.
\item Investigate the support properties of $\phi_t$ and conclude that $\Omega_t$ vanishes at $(\mathbf{x}_1,...,\mathbf{x}_N)$.
\end{enumerate}
\textit{Step 1:} We introduce the abbreviation for $t \in \R$
\begin{equation} \label{eq:defineft}
f_t:= \left[ \varphi_B(t), \varphi_A^\varepsilon(t_A) \right] U_B(t,s) \zeta 
\end{equation}
and recognize that the function $f: \R \to \mathscr{H}, t \mapsto f_t$ is bounded and measurable. Define
\begin{equation} \label{eq:definedphit}
\phi_t := \int_s^t d\tau \ e^{i\mathcal{H}_f(t-s)} e^{-i (\mathcal{H}_f + \mathcal{H}_B(s))(t-\tau)} e^{-i\mathcal{H}_f (\tau-s)} f_\tau.
\end{equation}
For $\eta \in \mathscr{D}, t \in \R$, we compute using Fubini's theorem,
\begin{equation} \label{eq:phiderivativeeta} \begin{split}
i \partial_t \left\langle \eta, \phi_t \right\rangle  = & i \partial_t \int^t_s d\tau \left\langle e^{i\mathcal{H}_f (\tau-s)} e^{i (\mathcal{H}_f + \mathcal{H}_B(s))(t-\tau)}e^{-i\mathcal{H}_f (t-s)}\eta, f_\tau \right\rangle
\\  = & \left\langle e^{i\mathcal{H}_f (t-s)} e^{i (\mathcal{H}_f + \mathcal{H}_B(s))(t-t)}e^{-i\mathcal{H}_f (t-s)}\eta, f_t \right\rangle
\\ & + \int_s^t d\tau \left\langle e^{i\mathcal{H}_f (\tau-s)} \mathcal{H}_B(s) e^{i (\mathcal{H}_f + \mathcal{H}_B(s))(t-\tau)}e^{-i\mathcal{H}_f (t-s)}\eta, f_\tau \right\rangle
\\ = & \left\langle \eta, f_t \right\rangle + \left\langle \mathcal{H}_B(t) \eta,  \phi_t  \right\rangle.
\end{split} \end{equation}
\textit{Step 2:} A calculation similar to the one above is now possible for $\Omega_t, t \in \R$:
\begin{equation}
 \begin{split}
i \partial_t \left\langle \eta, \Omega_t \right\rangle  = & i \partial_t \left( \left\langle U_B(s,t) \eta, \varphi_A^\varepsilon(t_A) \zeta \right\rangle - \left\langle \varphi_A^\varepsilon(t_A)^\dagger \eta,  U_B(t,s) \zeta \right\rangle \right)
\\  = & \left\langle U_B(s,t) \mathcal{H}_B(t) \eta, \varphi_A^\varepsilon(t_A) \zeta \right\rangle - \left\langle \varphi_A^\varepsilon(t_A)^\dagger \eta, \mathcal{H}_B(t) U_B(t,s) \zeta \right\rangle
\\ & - \left\langle \mathcal{H}_B(t) \eta, \varphi^\varepsilon_A(t_A) U_B(t,s) \zeta \right\rangle + \left\langle \mathcal{H}_B(t) \eta, \varphi^\varepsilon_A(t_A) U_B(t,s) \zeta \right\rangle
\\  = & \left\langle \mathcal{H}_B(t) \eta, \Omega_t \right\rangle + \left\langle \eta, \left[ \mathcal{H}_B(t), \varphi^\varepsilon_A(t_A) \right] U_B(t,s) \zeta \right\rangle 
\\ = & \left\langle \mathcal{H}_B(t) \eta, \Omega_t \right\rangle + \left\langle \eta, \left[ \varphi_B(t), \varphi^\varepsilon_A(t_A) \right] U_B(t,s) \zeta \right\rangle +  \sum_{k \in B} \left\langle \eta, \underbrace{\left[ \mathcal{H}^0_k, \varphi^\varepsilon_A(t_A) \right]}_{=0} \zeta \right\rangle
\\ = & \left\langle \eta,f_t\right\rangle + \left\langle \mathcal{H}_B(t) \eta, \Omega_t \right\rangle .
\end{split} \end{equation}
This together with \eqref{eq:phiderivativeeta} yields that the difference $\phi_t - \Omega_t$ is a weak solution of the Dirac equation in the sense that $\forall \eta \in \mathscr{D}$:
\begin{equation} \label{eq:Diracweak}
i \partial_t \left\langle \eta, \phi_t - \Omega_t \right\rangle =  \left\langle \mathcal{H}_B(t) \eta, \phi_t - \Omega_t \right\rangle \quad .
\end{equation}  
\textit{Step 3:} For all $s \in \R$, $U_B(s,s) = \id$ implies $\Omega_s=0$ and by definition, $\phi_s=0$. To show that $\Omega_t$ and $\phi_t$ are actually equal for all times $t \in \R$, it thus suffices to prove uniqueness of solutions to Eq.\ \eqref{eq:Diracweak}. 
\\ To this end, let $\rho: \R \to \mathcal{H},  t \mapsto \rho_t$ be continuous and for every $\eta \in \mathscr{D}$ a solution to 
\begin{equation} \label{eq:einundneunzig}
i \partial_t \left\langle \eta, \rho_t \right\rangle =  \left\langle \mathcal{H}_B(t) \eta, \rho_t \right\rangle.
\end{equation}
We claim that then, for all $t \in \R$, $\rho_t = U_B(t,s) \rho_s$. To see this we consider $t \mapsto \left\langle U_B(t,s) \eta,  \rho_t \right\rangle$, we prove that this is differentiable with zero derivative. For $h>0$, we find
\begin{equation} \begin{split}
& \ \frac{1}{h} \Big\|   \left\langle U_B(t+h,s) \eta,  \rho_{t+h} \right\rangle - \left\langle U_B(t,s) \eta,  \rho_{t} \right\rangle \Big\| 
\\ \leq & \ \Big\| \frac{1}{h} \left\langle U_B(t+h,s)\eta -U_B(t,s) \eta, \rho_{t+h} \right\rangle  -  \left\langle i \mathcal{H}_B(t) U_B(t,s) \eta,  \rho_{t+h} \right\rangle \Big\| \\ & +  \Big\| \frac{1}{h} \left\langle U_B(t,s) \eta, \rho_{t+h}- \rho_{t} \right\rangle - i \left\langle \mathcal{H}_B(t) U_B(t,s) \eta, \rho_{t+h} \right\rangle \Big\|
\\ \leq & \ \Big\| \frac{1}{h} \big( U_B(t+h,s)\eta -U_B(t,s) \eta \big) - i \mathcal{H}_B(t) U_B(t,s) \eta \Big\| \| \rho_{t+h} \| \\  & +  \Big\| \frac{1}{h} \left\langle U_B(t,s) \eta, \rho_{t+h}- \rho_{t} \right\rangle - i \left\langle \mathcal{H}_B(t) U_B(t,s) \eta, \rho_{t} \right\rangle \Big\| + \Big\| \left\langle \mathcal{H}_B(t) U_B(t,s) \eta, \rho_{t+h} - \rho_{t} \right\rangle \Big\| .
\end{split}
\end{equation}
The first term goes to zero as $h \rightarrow 0$ because $\eta \in \mathscr{D}$ and since $\rho_t$ is continuous, the norm $\rho_{t+h}$ is bounded in a neighbourhood of $t$. The second term vanishes using \eqref{eq:einundneunzig}, noting that also $U_B(t,s)\eta \in \mathscr{D}$ by Corollary \ref{thm:corollaryDinvariant}. The last term also goes to zero by continuity of $\rho_t$. We have thus proven that 
\begin{equation}
\partial_t \left\langle U_B(t,s) \eta, \rho_t \right\rangle = 0 \ \Rightarrow  \left\langle \eta, U_B(s,t) \rho_t \right\rangle = const.
\end{equation}
This implies the desired uniqueness statement $\left\langle \eta, U_B(t,s) \rho_s - \rho_t \right\rangle = 0$ for all $\eta \in \mathscr{D}$. Since $\mathscr{D} \subset \mathscr{H}$ is dense, $\rho_t = U_B(t,s) \rho_s$ follows. 
\\ In the special case of \eqref{eq:Diracweak}, the initial value is $\rho_s=\phi_s-\Omega_s=0$. Furthermore, $t \mapsto \Omega_t - \phi_t$ is continuous, hence  
\begin{equation}
\forall t \in \R: \phi_t - \Omega_t = 0.
\end{equation}
\textit{Step 4:} Thanks to Eq.\ \eqref{eq:definedphit}, we now have an explicit formula for $\Omega_t$ by means of $\Omega_t=\phi_t$. Next, we investigate its support.
\\ To treat the commutator term in \eqref{eq:defineft}, we insert two identities:
\begin{equation} \begin{split}
  \left[ \varphi_B(t),\varphi_A^\varepsilon(t_A) \right] = & \ \frac{1}{1 + i \varepsilon \varphi_A(t_A)} (1 + i \varepsilon \varphi_A(t_A)) \varphi_B(t) \varphi_A(t_A) \frac{1}{1 + i \varepsilon \varphi_A(t_A)} 
\\  & -  \ \frac{1}{1 + i \varepsilon \varphi_A(t_A)} \varphi_A(t_A) \varphi_B(t) (1 + i \varepsilon \varphi_A(t_A))\frac{1}{1 + i \varepsilon \varphi_A(t_A)} 
\\ = &  \ \frac{1}{1 + i \varepsilon \varphi_A(t_A)} \left[ \varphi_B(t), \varphi_A(t_A) \right] \frac{1}{1 + i \varepsilon \varphi_A(t_A)}.
\end{split}
\end{equation}
The operator $\frac{1}{1 + i \varepsilon \varphi_A(t_A)}$ does not increase the domain of functions since it is the resolvent of $\varphi_A(t_A)$ that can be written as a direct fiber integral, compare \cite[thm.\ 3.4]{zenkstockmeyer} and \cite[thm. XIII.85]{reedsimon4}. Hence, Lemma \ref{thm:lemmacommutators} guarantees that $f_t(\mathbf{x}_1,...,\mathbf{x}_N)=0$ whenever $
\| \mathbf{x}_j - \mathbf{x}_k \| > \delta + |t - t_A|$ for all $j \in A, k \in B$. 

The spatial support is not altered by the $\mathcal{H}_f$ operators and their exponentials, so we have
\begin{equation}
\supp \left(e^{-i \mathcal{H}_f(\tau-s)} f_\tau \right) \subset \left\{ (\mathbf{x}_1,...,\mathbf{x}_N) \in \R^{3N} \big| \exists j \in A, k \in B: \| \mathbf{x}_j - \mathbf{x}_k \| \leq \delta + \tau - t_A  \right\}.
\end{equation}
Applying Lemma \ref{thm:lemmacausal}, this support can grow by at most $\sum_{k \in B} B_{t-\tau}(\mathbf{x}_j) $ when acted on by $e^{-i( \mathcal{H}_f + \mathcal{H}_B(s))(t-\tau)}$. So this implies
\begin{equation} \label{eq:drueber}
\supp \left(e^{-i( \mathcal{H}_f + \mathcal{H}_B(s))(t-\tau)} e^{-i \mathcal{H}_f(\tau-s)} f_\tau \right) \subset \left\{ (\mathbf{x}_1,...,\mathbf{x}_N) \in \R^{3N} \Big| \begin{array}{c}\exists  j \in A, k \in B: \\
\| \mathbf{x}_j - \mathbf{x}_k \| \leq \delta + t- t_A  \end{array} \right\}. 
\end{equation}
Consider $\Omega_{t_B}=\phi_{t_B}$. By \eqref{eq:drueber}, the integrand in Eq.\ \eqref{eq:definedphit} vanishes whenever $\| \mathbf{x}_j - \mathbf{x}_k \| > \delta + t_B- t_A$. This is satisfied for $(\mathbf{x}_1,...,\mathbf{x}_N)$ by assumption, which yields
\begin{equation}
\Omega_{t}(\mathbf{x}_1,...,\mathbf{x}_N) = \left( \left[ U_B(t,s), \varphi^\varepsilon_A(t_A) \right] \zeta \right) (\mathbf{x}_1,...,\mathbf{x}_N) = 0
\end{equation}
for every positive $\varepsilon$, and thus with \eqref{eq:somethingfollows} the claim of the lemma. \qed
\end{proof}

We are now ready to prove the existence Theorem \ref{thm:mainthm}. In addition to the claim in Thm.\  \ref{thm:mainthm} we also prove the following extended claim that states the form of the solution.

\begin{theorem} \label{thm:consistencydeltafree}
For each $\psi^0 \in \mathscr{D}$, there exists a solution $\psi$ of the multi-time system in the sense of Def.\ \ref{def:solution} on $\mathscr{S}_\delta$ with initial data $\psi(0, ... ,0) = \psi^0$ and with $\psi(t_1,...,t_N) \in \mathscr{D}$. 
\\ Let $\sigma$ be a permutation on $\{1,...,N\}$ such that $t_{\sigma(1)} \geq t_{\sigma(2)} \geq \dots \geq t_{\sigma(N)}$, then one such solution is given by 
\begin{align} \label{eq:solutiond}
& \psi(t_1,...,t_N) 
\\  &= U_{\{ \sigma(1) \}}(t_{\sigma(1)}, t_{\sigma(2)}) \dots U_{\{ \sigma(1),...,\sigma(N-1) \} } (t_{\sigma(N-1)},t_{\sigma(N)})U_{\{ 1,2,...,N \} } (t_{\sigma(N)},0) \psi^0 . \nonumber
\end{align}
\end{theorem}

For the proof, it will be helpful to abbreviate formulas like \eqref{eq:solutiond} using the $\bigcirc$-symbol for the ordered product of operators, $\bigcirc_{k=1}^l A_k := A_1 A_2 ... A_l$.
In this notation, expression \eqref{eq:solutiond} reads
\begin{equation}
 \left( \left( \bigcirc_{k=1}^{N-1} U_{\{ \sigma(j) | j \leq k \}}(t_{\sigma(k)},t_{\sigma(k+1)} ) \right) U_{\{ 1,,...,N \} } (t_{\sigma(N)},0) \psi^0 \right) (\mathbf{x}_1,...,\mathbf{x}_N). 
\end{equation}
Compare also fig.\ \ref{fig:Us} for a depiction of the successive application of the $U_A$ operators in a simple case.

\begin{proof}
Let $\psi^0 \in \mathscr{D}$, and define $\psi: \R^{N} \to \mathscr{H}$ by Eq.\ \eqref{eq:solutiond}. Property $U_A(t,t) = \id$ stated in Theorem \ref{thm:existsU} ensures $\psi(0,...,0)=\psi^0$, so the correct initial value is attained. $\psi^0 \in \mathscr{D}$ implies that for all $t_1,...,t_N \in \R$, ${\psi}(t_1,...,t_N) \in \mathscr{D}$ since $\mathscr{D}$ is preserved by the operators $U_A$ by virtue of Corollary \ref{thm:corollaryDinvariant}. 

We now show the three points from Definition \ref{def:solution}. 

\emph{i)} Since $\psi: \R^N \to \mathscr{D} \subset \mathscr{H}$, we may infer by Theorem \ref{thm:existsU} part 3 that $\psi$ is differentiable.

\emph{ii)} Let $j \in \{1,...,N\}$. By Lemma \ref{thm:invariantdomainforops} also $\mathcal{H}_j(t_j) \psi(t_1,...,t_N) \in \mathscr{D}$, so both expressions are pointwise evaluable. The same is true for $\partial_{t_j} \psi(t_1,...,t_N)$ since it amounts to a successive application of $U_A$ operators and of $\mathcal{H}_j$, which all leave $\mathscr{D}$ invariant.

\emph{iii)} We now have to check that $\psi$ satisfies the respective equations \eqref{eq:systemstrongsolution} in $S_\delta$. Given a set $A \subset \{ 1,...,N\}$ and a time $t_A \in \R$, consider a configuration $(t_1, \mathbf{x}_1,...,t_N,\mathbf{x}_N) \in \mathscr{S}_\delta$ where $t_j = t_A ~ \forall j \in A$. We assume w.l.o.g.\ that the times are already ordered $t_1 \geq t_2 \geq \dots \geq t_N$, so that the permutation in \eqref{eq:solutiond} is the identity. Let $a := \min (A)$ and $b := \max (A)$, then
\begin{align}
 \psi(t_1,...,t_N) \nonumber
 = &  \left( \bigcirc_{k=1}^{a-2} U_{\{ j | j \leq k \}}(t_{k},t_{k+1} ) \right) U_{\{ j | j \leq a-1 \}}(t_{a-1},t_A ) U_{\{j| j \leq b \}}(t_A, t_{b + 1})  \nonumber
\\ &  \left( \bigcirc _{k=b+1}^{N-1} U_{\{ j | j \leq k \}}(t_{k},t_{k+1} ) \right) U_{\{ 1,,...,N \} }  (t_{N},0) \psi^0  \label{eq:this}
\end{align}
We take the derivative of \eqref{eq:this} with respect to $t_A$ and use that  for $\zeta \in \mathscr{D}$,
\begin{equation}
i \frac{d}{dt} U_B(s,t)\zeta = - U_B(s,t) \mathcal{H}_B(t) \zeta, \quad \forall \ s,t \in \R, B \subset \{1,...,N\},
\end{equation}
which follows directly from the properties of the time evolution operators. Abbreviating 
\begin{equation}
\psi' := U_{\{j| j \leq b \}}(t_A, t_{b + 1}) \left( \bigcirc _{k=b+1}^{N-1} U_{\{ j | j \leq k \}}(t_{k},t_{k+1} ) \right) U_{\{ 1,,...,N \} }  (t_{\sigma(N)},0) \psi^0,
\end{equation}
we obtain
\begin{equation}
\begin{split}
& i \frac{\partial}{\partial t_A} \psi(t_1,...,t_N)   \\ = ~ & \left( \left( \bigcirc_{k=1}^{a-2} U_{\{ j | j \leq k \}}(t_{k},t_{k+1} ) \right) U_{\{ j | j \leq a-1 \}}(t_{a-1},t_A ) \right.  
 \left. \left( -\mathcal{H}_{\{ j | j \leq a-1 \}}(t_A) + \mathcal{H}_{\{ j | j \leq b \}}(t_A) \right)  \psi' \right) 
\\ = ~ & \mathcal{H}_A(t_A) \psi (t_1,...,t_N)  
 + \left( \left[  \bigcirc_{k=1}^{a-1} U_{\{ j | j \leq k \}}(t_{k},t_{k+1} ) , \mathcal{H}_A(t_A) \right] \psi' \right) . \label{eq:126}
\end{split} \end{equation}
We rewrite the second term as
\begin{equation}
\begin{split}
&  \left[  \bigcirc_{k=1}^{a-1} U_{\{ j | j \leq k \}}(t_{k},t_{k+1} ) , \mathcal{H}_A(t_A) \right] \psi' 
\\ = & \sum_{l=1}^{a-1} \left( \bigcirc_{k=1}^{l-1} U_{\{ j | j \leq k \}}(t_{k},t_{k+1})  \right) \left[ U_{\{ j | j \leq l \}}(t_{l},t_{l+1}), \mathcal{H}_A(t_A) \right] \left( \bigcirc_{k=l+1}^{a-1} U_{\{ j | j \leq k \}}(t_{k},t_{k+1})  \right) \psi',
\end{split} \end{equation}
where empty products such as $\bigcirc_{k=1}^{0}$ denote $\id$. Lemma \ref{thm:subclaim} implies that for any $\zeta \in \mathscr{D}$ and $l < a$, 
\begin{equation}
\supp \left( \left[ U_{\{ j | j \leq l \}}(t_l,t_{l+1}), \mathcal{H}_A(t_A)  \right] \zeta \right) \subset \{ (\mathbf{x}_1,...,\mathbf{x}_N) | \exists  k \in A, j \leq l: \| \mathbf{x}_j - \mathbf{x}_k \| \leq \delta + t_l - t_A  \}.
\end{equation}

The support properties of the evolution operators (Lemma \ref{thm:lemmacausal}) imply that if $\supp (\xi) \subset R$, then $\supp \left( \bigcirc_{k=1}^{l-1} U_{\{ j | j \leq k \}}(t_{k},t_{k+1}) \xi \right)$ is a subset of
\begin{equation}
 \left\{ (\mathbf{y}_1,...,\mathbf{y}_N) \in \R^{3N} \Big| \exists  (\mathbf{x}_1,...\mathbf{x}_N) \in R: \begin{array}{c} \mathbf{x}_j = \mathbf{y}_j \ \text{if} \ j > l. 
\\ \| \mathbf{x}_j - \mathbf{y}_j \| \leq t_j - t_l \ \text{if}  \ j \leq l.
\end{array} \right\} \label{eq:supportargument}
\end{equation}

Now we see that the support growth described by \eqref{eq:supportargument} is exactly such that the term $\left[  \bigcirc_{k=1}^{a-1} U_{\{ j | j \leq k \}}(t_{k},t_{k+1} ) , \mathcal{H}_A(t_A) \right] \psi' (\mathbf{x}_1,...,\mathbf{x}_N)=0,$ whenever $\| \mathbf{x}_j - \mathbf{x}_k \| > \delta + |t_j - t_k|$ holds for all $j \in A, k \notin A$. Thus \eqref{eq:126} evaluated inside of $\mathscr{S}_\delta$ becomes
\begin{equation}
 \left( i \frac{\partial}{\partial t_A} \psi(t_1,...,t_N)  \right) (\mathbf{x}_1,...,\mathbf{x}_N) = \left( \mathcal{H}_A(t_A) \psi (t_1,...,t_N)\right)(\mathbf{x}_1,...,\mathbf{x}_N) , 
\end{equation}
which proves that $\psi$ indeed is a solution of the multi-time system \eqref{eq:systemstrongsolution}. \qed
\end{proof}

\subsection{Uniqueness of solutions} \label{sec:uniqueness}

Uniqueness of solutions can be proven by induction over the particle number, using the key features of our multi-time system that the Hamiltonians $\mathcal{H}_k$ are self-adjoint and that the propagation speed is bounded by the speed of light (see Lemma \ref{thm:lemmacausal}).
 
\begin{proofU}
Let $\psi_1$, $\psi_2$ be solutions to \eqref{eq:systemstrongsolution} in the sense of Def.\ \ref{def:solution} with $\psi_1(0,...,0) = \psi_2(0,...,0) = \psi^0$. Due to linearity, $\omega := \psi_1 - \psi_2$ is a solution to \eqref{eq:systemstrongsolution} in the sense of Def.\ \ref{def:solution} with initial value $\omega(0,...,0) = \psi^0 - \psi^0= 0$. In particular, the point-wise evaluations of $\omega$ as in \eqref{eq:pointwiseev} are also well-defined. By induction over $L \in \{1,...,N\}$, we prove the statement: 
\\ $\mathbf{A(L)}$: \emph{At all points $(t_1,\mathbf{x}_1,...,t_N,\mathbf{x}_N) \in \mathscr{S}_\delta$ with at most $L$ different time coordinates, we have
$\left( \omega (t_1,...,t_N) \right) (\mathbf{x}_1,...,\mathbf{x}_N) = 0$.}
\\ For the \textbf{base case} $\mathbf{A(1)}$, we consider configurations with all times equal, where $\omega$ satisfies
\begin{equation}
i \partial_t \omega(t,...,t) = \mathcal{H}_{\{1,...,N\}}(t) \omega (t,...,t).
\end{equation}
By the uniqueness statement in Theorem \ref{thm:existsU}, this implies 
\begin{equation}
\omega(t,...,t)  = U_{\{1,...,N\}}(t,0) \omega^0= 0.
\end{equation}
$\mathbf{A(L) \Longrightarrow A(L+1)}$: We assume that $A(L)$ holds, and let $(t_1,\mathbf{x}_1,...,t_N,\mathbf{x}_N) \in \mathscr{S}_\delta$ with exactly $L+1$ different time coordinates. This means there is a unique partition of $\{1,...,N\}$ into disjoint sets $\Pi_1,...,\Pi_{L+1}$ which groups together particles with the same time coordinate in an ascending way:
\begin{equation} \begin{split}
\Pi_1 & \ := \left\{ j \in \{1,...,N\} \Big| t_j = \min_{k \in \{1,...,N\}} t_k \right\} \\
\Pi_2 & \ := \left\{ j \in \{1,...,N\} \Big| t_j = \min_{k \in \{1,...,N\} \setminus \Pi_1} t_k \right\}
\\ & \cdots
\\ \Pi_m & \ := \left\{ j \in \{1,...,N\} \Big| t_j = \min_{k \in \{1,...,N\} \setminus \cup_{i=1}^{m-1}\Pi_i} t_k \right\}.
 \end{split}
\end{equation}
Denote the largest time by $t_{L+1}$ and the second largest one by $t_L$. We define the backwards lightcone with respect to the particles in $\Pi_{L+1}$ as follows,
\begin{equation}
B:=\left\{ (y_1,...,y_N) \in \R^{4N} \left| \begin{array}{c} y_j = x_j \ \text{if} \ j \notin \Pi_{L+1} 
\\  \forall j \in \Pi_{L+1}: y^0_j = \tau \ \text{with} \  t_L \leq \tau \leq t_{L+1}, \\ \ |\mathbf{y}_j - \mathbf{x}_j| \leq t_{L+1} - \tau  
\end{array}  \right. \right\}.
\end{equation}
We show that $B \subset \mathscr{S}_{\delta}$. If $(y_1,...,y_N) \in B$, consider $j \in \Pi_{L+1}$ and $k \notin \Pi_{L+1}$, then 
\begin{equation} \begin{split}
|y^0_k-y^0_j|+\delta & = \tau - t_k + \delta =  (\tau - t_{L+1}) + (t_{L+1} - t_k + \delta) 
\\ & < -|\mathbf{y}_j - \mathbf{x}_j| + |\mathbf{x}_k - \mathbf{x}_j| \leq |\mathbf{x}_k - \mathbf{y}_j| = |\mathbf{y}_k - \mathbf{y}_j| .
 \end{split}
\end{equation}
Thus, all points in $B$ are still in our domain $\mathscr{S}_\delta$. In particular, we have
\begin{equation}
\left( i \partial_\tau  \omega(y^0_1,...,y^0_N) \right) (\mathbf{y}_1,...,\mathbf{y}_N) = \left( \mathcal{H}_{\Pi_{L+1}}(\tau) \omega(y^0_1,...,y^0_N ) \right) (\mathbf{y}_1,...,\mathbf{y}_N)  \quad \forall (y_1,...,y_N) \in B.
\end{equation}
Since $B$ contains the domain of dependence, i.e.\ the set that uniquely determines the value of $\omega$ at $(t_1,\mathbf{x}_1,...,t_N,\mathbf{x}_N)$ according to Lemma \ref{thm:lemmacausal}, Theorem \ref{thm:existsU} tells us that 
\begin{equation}
\omega(x_1,...,x_N) = \left( U_{\Pi_{L+1}}(t_{L+1},t_L) \omega^{t_L} \right) (\mathbf{x}_1,...,\mathbf{x}_N),
\end{equation}
where $\omega^{t_L}$ denotes the function $\omega$ evaluated at the time coordinates as in $(t_1,...,t_N)$ but where $t_{L+1}$ is replaced by $t_L$. This only has $L$ different times and is thus given according to the induction hypothesis ${A(L)}$ as $\omega^{t_L}=0$ in the whole domain of dependence. Consequently,
\begin{equation}
\left( \omega (t_1,...,t_N) \right) (\mathbf{x}_1,...,\mathbf{x}_N) = 0,
\end{equation}
which concludes the uniqueness proof.
 \qed 
\end{proofU}

\subsection{Interaction} \label{sec:interactionofdfp}

We now demonstrate that our model is indeed interacting, providing a rigoros version of Eq.\ \eqref{eq:ehrenfestdelta}.

\begin{proof}[Proof of Theorem~\ref{thm:interaction}] Let $t \in \R$ and $\mathbf{x} \in \R^3$.
The first step just uses that $\psi^t$ solves the Dirac equation,
\begin{equation}
i \partial_t \left\langle \psi^t, \varphi(t,\mathbf{x}) \psi^t \right\rangle = \left\langle - \mathcal{H}^t \psi^t, \varphi(t,\mathbf{x}) \psi^t \right\rangle + \left\langle   \psi^t, \varphi(t,\mathbf{x}) \mathcal{H}^t \psi^t \right\rangle + \left\langle \psi^t, i\dot{\varphi}(t,\mathbf{x}) \psi^t \right\rangle .
\end{equation}
We already encountered $\dot{\varphi}$, the time-derivative of the operator $\varphi$, in the proof of Lemma \ref{thm:katorellichHtilde}.
Since $\mathcal{H}^t$ and $\varphi(t,\mathbf{x})$ commute at equal times, only the third summand survives and the second derivative is 
\begin{equation} \begin{split}
\partial_t^2 \left\langle \psi^t, \varphi(t,\mathbf{x}) \psi^t \right\rangle & = -i \partial_t \left\langle \psi^t, i\dot{\varphi}(t,\mathbf{x}) \psi^t \right\rangle 
\\ & =  i \left\langle  \mathcal{H}^t \psi^t, \dot{\varphi}(t,\mathbf{x}) \psi^t \right\rangle - i \left\langle   \psi^t, \dot{\varphi}(t,\mathbf{x}) \mathcal{H}^t \psi^t \right\rangle + \left\langle  \psi^t, \ddot{\varphi}(t,\mathbf{x}) \psi^t \right\rangle 
\\ & = i \left\langle  \psi^t,  \left[ \mathcal{H}^t, \dot{\varphi}(t,\mathbf{x}) \right] \psi^t \right\rangle + \left\langle  \psi^t, \triangle_{\mathbf{x}} \varphi(t,\mathbf{x}) \psi^t \right\rangle .
\end{split} \end{equation}
Hence,
\begin{equation} \label{eq:felix}
\square \left\langle \psi^t, \varphi(t,\mathbf{x}) \psi^t \right\rangle = i \left\langle  \psi^t,  \left[ \mathcal{H}^t, \dot{\varphi}(t,\mathbf{x}) \right] \psi^t \right\rangle = i \sum_{k=1}^N \left\langle  \psi^t,  \left[ \varphi_k(t), \dot{\varphi}(t,\mathbf{x}) \right] \psi^t \right\rangle.
\end{equation}
So we need to compute, with the integration variable $x=(\mathbf{x}_1,...,\mathbf{x}_N)$,
\begin{equation} \begin{split}
&  i \left\langle  \psi^t,  \left[ \varphi_k(t), \dot{\varphi}(t,\mathbf{v}) \right] \psi^t \right\rangle  \\ = \ & i \int d^{3N}\hspace*{-0.1cm}x \ \psi^{t \dagger}(x) \int  \frac{d^3\mathbf{k}}{2\omega(\mathbf{k})}\left( \hat{\rho}^\dagger(\mathbf{k}) \hat{\rho}(\mathbf{k}) i \omega(\mathbf{k}) e^{-i\mathbf{k}(\mathbf{x}_k-\mathbf{v})}- c.c. \right) \psi^t(x)
\\ = \ & - \frac{1}{2}  \int d^{3N}\hspace*{-0.1cm}x \ \psi^{t \dagger}(x) \int  d^3\mathbf{k} \ \hat{\rho}^\dagger(\mathbf{k}) \hat{\rho}(\mathbf{k})  \left( e^{-i\mathbf{k}(\mathbf{x}_k-\mathbf{v})}+ e^{i\mathbf{k}(\mathbf{x}_k-\mathbf{v})} \right) \psi^t(x).
\end{split} \end{equation}
Denoting the function $\mathbf{y} \mapsto \rho(\mathbf{y}+\mathbf{v}-\mathbf{x}_k)$ by $\omega$, we have $\hat{\omega}(\mathbf{k}) = \hat{\rho}(\mathbf{k}) e^{i\mathbf{k}(\mathbf{x}_k-\mathbf{v})}$. Thus, the above formula can be rewritten with the help of the Plancherel theorem,
\begin{equation} \begin{split}
& \ = - \frac{1}{2} \int d^{3N}\hspace*{-0.1cm}x \ \psi^{t \dagger}(x) \left( \left\langle \hat{\rho}, \hat{\omega} \right\rangle_{L^2(\R^3)} + \left\langle \hat{\omega}, \hat{\rho} \right\rangle_{L^2(\R^3)} \right) \psi^t(x)
\\ & \ = - \frac{1}{2} \int d^{3N}\hspace*{-0.1cm}x \ \psi^{t \dagger}(x) \left( \left\langle \rho, \omega \right\rangle_{L^2(\R^3)} + \left\langle \omega, \rho \right\rangle_{L^2(\R^3)} \right) \psi^t(x)
\\ & \ = - \int d^{3N}\hspace*{-0.1cm}x \ \psi^{t \dagger}(x)  \left\langle \rho, \omega \right\rangle_{L^2(\R^3)}  \psi^t(x).
\\ & \ = -  \int d^{3N}\hspace*{-0.1cm}x \ \psi^{t \dagger}(x) \int d^3\mathbf{y}_1  {\rho}(\mathbf{y}_1) \rho(\mathbf{v} - \mathbf{x_k} + \mathbf{y}_1)  \psi^t(x)
\end{split} \end{equation}
We have used that $\rho$ and $\omega$ are real-valued. The result contains the term we wrote as $\rho ** \delta(\mathbf{x}_k-\mathbf{v})$ in \eqref{eq:deltastarstar}. Inserting this into \eqref{eq:felix} gives
\begin{equation} \begin{split}
\square \left\langle \psi^t, \varphi(t,\mathbf{x}) \psi^t \right\rangle & = - \sum_{k=1}^N \int d^{3N}\hspace*{-0.1cm}x \ \psi^{t \dagger}(x) \int d^3\mathbf{y}_1  {\rho}(\mathbf{y}_1) \rho(\mathbf{x} - \mathbf{x_k} + \mathbf{y}_1)  \psi^t(x)
 \\ & \equiv  - \sum_{k=1}^N \left\langle \psi^t, \rho ** \delta(\hat{\mathbf{x}}_k - \mathbf{x}) \psi^t \right\rangle,
\end{split} \end{equation}
which concludes the proof. \qed \end{proof}

 \section*{Acknowledgements} The authors wish to thank Heribert Zenk and Roderich Tumulka for helpful exchange as well as Lea Bo{\ss}mann, Robin Schlenga and Felix H{\"a}nle for discussions, support and encouragement. This work was partly funded by the Elite Network of Bavaria through the Junior Research Group \textit{Interaction between Light and Matter}. L.N.\ gratefully acknowledges financial support by the \textit{Studienstiftung des deutschen Volkes}.

\end{document}